\documentclass[12pt,draftcls,onecolumn]{IEEEtran}                                                          % paper

\IEEEoverridecommandlockouts                              % This command is only
\usepackage{graphics} % for pdf, bitmapped graphics files
\usepackage{epsfig} % for postscript graphics files
\usepackage{mathptmx} % assumes new font selection scheme installed
\usepackage{times} % assumes new font selection scheme installed
\usepackage{amsmath} % assumes amsmath package installed
\usepackage{amssymb}  % assumes amsmath package installed
\usepackage{floatrow}
\usepackage{multirow}
\newtheorem{theorem}{\indent Theorem}
\newtheorem{lemma}[theorem]{\indent Lemma}
\newtheorem{corollary}[theorem]{\indent Corollary}
\newtheorem{proposition}[theorem]{\indent Proposition}
\newtheorem{definition}[theorem]{\indent Definition}
\newtheorem{example}{\indent Example}
\newtheorem{remark}{\indent Remark}

\title{\LARGE \bf
Sampled-data design for robust control of a single qubit }

\author{Daoyi~Dong, Ian R. Petersen and Herschel Rabitz% <-this % stops a space
\thanks{This work was supported by the Australian Research Council. H.R. acknowledges support from the ARO.}% <-this % stops a space
\thanks{D. Dong is with the School of
Engineering and Information Technology, University of New South
Wales at the Australian Defence Force Academy, Canberra, ACT 2600,
Australia
and the Institute of Cyber-Systems and Control, State Key Laboratory of Industrial Control Technology, Zhejiang University, Hangzhou 310027, China
        {\tt\small daoyidong@gmail.com}}%
\thanks{I. R.
Petersen is with the School of Engineering and Information
Technology, University of New South Wales at the Australian Defence
Force Academy, Canberra, ACT 2600, Australia
        {\tt\small i.r.petersen@gmail.com}}
        \thanks{H.
Rabitz is with the Department of Chemistry, Princeton University, Princeton, New Jersey 08544, USA
        {\tt\small hrabitz@princeton.edu}}%
}

\begin{document}

\maketitle
\thispagestyle{empty}
\pagestyle{empty}

%%%%%%%%%%%%%%%%%%%%%%%%%%%%%%%%%%%%%%%%%%%%%%%%%%%%%%%%%%%%%%%%%%%%%%%%%%%%%%%%
\begin{abstract}

This paper presents a sampled-data approach for the robust control of a single qubit (quantum bit).
The required robustness is defined using a sliding
mode domain and the control law is designed offline and then utilized
online with a single qubit having bounded uncertainties. Two classes of uncertainties are considered involving the system Hamiltonian and the coupling strength of the system-environment interaction. Four cases are analyzed in detail including without decoherence, with amplitude damping decoherence, phase damping decoherence and depolarizing decoherence. Sampling periods are specifically designed for these cases to guarantee the required robustness. Two sufficient conditions are presented for guiding the design of unitary control for the cases without decoherence and with amplitude damping decoherence. The proposed approach has potential applications in quantum error-correction and in constructing robust quantum gates.

\end{abstract}


\newpage
%%%%%%%%%%%%%%%%%%%%%%%%%%%%%%%%%%%%%%%%5
\section*{Nomenclature}
Throughout this paper, we use the following notation:\\
$|\psi\rangle$ \ \ \ \ \ \ \ \ state vector (quantum pure state)\\
$a^{*}$ \ \ \ \ \ \ \ \ \ complex conjugate of $a$\\
$L^{T}$ \ \ \ \ \ \ \ \ \ transpose of $L$\\
$L^{\dagger}$ \ \ \ \ \ \ \ \ \ adjoint of $L$\\
$\text{tr}(A)$ \ \ \ \ \ \ trace of $A$\\
$\langle\psi|$ \ \ \ \ \ \ \ \ adjoint of $|\psi\rangle$\\
$\langle \phi | \psi\rangle$ \ \ \ \ \ inner product of
$|\phi\rangle$
and $|\psi\rangle$ \\
$\rho$ \ \ \ \ \ \ \ \ \ \ density operator\\
$\sigma_{x,y,z}$  \ \ \ \ \ \ Pauli matrices\\
$\omega(t)$ \ \  \  \ \ \ uncertainty amplitude in $\sigma_{z}$\\
$\epsilon_{x}(t)$ \ \ \ \ \ \ uncertainty amplitude in $\sigma_{x}$\\
$\epsilon_{y}(t)$ \  \ \ \ \ \ uncertainty amplitude in $\sigma_{y}$\\
$\mathbf{R}$ \ \ \ \ \ \ \ \ set of real numbers\\
$\gamma$   \ \ \ \ \ \ \ \ \ coupling strength\\
$\delta\gamma_{t}$   \ \ \ \ \ \ uncertainty in coupling strength\\
$\mathcal{D}_{c}$   \ \ \ \ \ \  sliding mode domain of closed systems\\
$\mathcal{D}_{a}$   \ \ \ \ \ \  sliding mode domain of quantum systems with amplitude damping decoherence\\
$\mathcal{D}_{p}$   \ \ \ \ \ \  sliding mode domain of quantum systems with phase damping decoherence\\
$\mathcal{D}_{d}$   \ \ \ \ \ \  sliding mode domain of quantum systems with depolarizing decoherence\\
$C$   \ \ \ \ \ \ \ \ coherence\\
$P$   \ \ \ \ \ \ \ \  purity\\
$p_{0}$   \ \ \ \ \ \ \ probability of failure\\
%$p_{01}$   \ \ \ \ \ \ \ \ error probability of measurement from $|0\rangle$ to $|1\rangle$\\
%$p_{10}$   \ \ \ \ \ \ \ \ error probability of measurement from $|1\rangle$ to $|0\rangle$\\

%%%%%%%%%%%%%%%%%%%%%%%%%%%%%%%%%%%%%%%%%%%%%%%%%%%%%%%%%%%%%%%%%%%%%%%%%%%%%%%%
\section{INTRODUCTION}\label{Sec1}

Controlling quantum phenomena is becoming an important task in
different research areas such as quantum optics, physical chemistry
and quantum information \cite{Dong and Petersen 2010IET}-\cite{Nielsen and Chuang 2000}. The development of quantum control theory can provide
systematic methods and a theoretical framework for analyzing and
synthesizing quantum control problems. Several theoretical tools and design methods in classical control have been applied to the quantum domain. For example, Lie groups and Lie algebras have been used to establish controllability conditions for closed quantum systems \cite{D'Alessandro 2007}. Optimal control theory has been applied to control analysis and the design of several quantum control tasks such as population transfer with minimum energy or in the shortest time \cite{Khaneja et al 2001}-\cite{Boscain and Mason
2006}. Learning control has become a powerful tool for the direct laboratory discovery of laser pulses controlling a variety of atomic and molecular phenomena \cite{Rabitz et al 2000}. Feedback control has been utilized for the control of quantum entanglement, quantum error-correction and quantum state preparation \cite{Wiseman and Milburn 1993}-\cite{Mirrahimi and van Handel 2007}. The development of quantum control theory needs to consider the special characteristics of quantum systems (e.g., measurement collapse and non-commutative relationships) and the unique objectives of quantum control (e.g., entanglement generation and decoherence control) (For more discussion, see, e.g., \cite{Dong and Petersen 2010IET}).

Robust control is one of the most important research areas in classical control theory. Attaining robust control for quantum systems has been recognized as a key issue in the development of
practical quantum technology \cite{Pravia et al 2003}-\cite{Dong et
al IJC}, since many types of uncertainties
are unavoidable (including control noise, environmental disturbances, etc.) for most
practical quantum systems. Several methods have been proposed for
the robust control of quantum systems. For example, James \emph{et
al.} \cite{James et al 2007} formulated and solved a quantum
robust control problem using the $H^{\infty}$ method for linear
quantum stochastic systems. A risk-sensitive control problem has been solved for a sampled-data feedback
model of quantum systems \cite{James 2004}. Quantum robust control is still in its infancy, and it is necessary to develop new tools to deal with different types of uncertainties.

Dong and Petersen \cite{Dong and Petersen 2009NJP}-\cite{Dong and Petersen 2011IFAC} developed sliding mode control to
enhance the robustness of quantum systems. In particular, two
approaches based on sliding mode design \cite{Utkin 1977} have been
proposed for the control of quantum systems, and potential applications of sliding mode control
to quantum information processing have been presented \cite{Dong and
Petersen 2009NJP}. Sliding
mode control for two-level quantum systems was presented to deal with
bounded uncertainties in the system Hamiltonian
\cite{Dong and Petersen 2011Automatica}. This paper will
employ the concept of a sliding mode domain to define the required robustness and develop a new sampled-data design
approach \cite{Dong and Petersen 2011CDC}, \cite{Chen and Francis 1995} to enhance the performance of
a controlled quantum system with uncertainties in the Hamiltonian as well as in the system-environment interaction.

Sampled-data control has been widely applied in industrial electronics, process control and signal processing \cite{Chen and Francis 1995}. The sampled data are used to design controllers while the sampling (measurement) process is usually assumed not to affect the system's state. However, in quantum control, the sampling process unavoidably destroys the system's state according to the measurement collapse postulate (see, e.g., \cite{Nielsen and Chuang 2000}). Hence, measurement can be used as the means for information acquisition as well as a control tool. For example, several incoherent control schemes have been presented where measurements are used as a control tool to affect the system dynamics \cite{Vilela Mendes and Man'ko 2003}-\cite{Romano and D'Alessandro 2006}. A framework of quantum operations including unitary control and projective measurements has been developed to investigate feedback control of quantum systems \cite{Belavkin 1983}, \cite{Bouten et al 2009}. One well known example where measurement modifies the system dynamics is the quantum Zeno effect, which is the inhibition of transitions between quantum
states by frequent measurement of the state (see, e.g.,
\cite{Misra and Sudarshan} and \cite{Itano et al 1990}). However, it is usually a difficult task
to make frequent measurements with practical quantum systems. We may assume
that the smaller the measurement period is, then the bigger the cost of accomplishing
the periodic measurements becomes. Hence, in contrast to
the quantum Zeno effect, in this paper we will use the sampling (projective measurement) process as a control tool and design sampling periods as large as possible to guarantee the required robustness for several classes of quantum control tasks including control design for quantum systems with uncertainties in the system Hamiltonian and robust decoherence control of Markovian open quantum systems.

Decoherence occurs when a quantum system interacts with an uncontrollable environment \cite{Breuer and Petruccione 2002}. Decoherence has been recognized as a bottleneck for the development of practical quantum information technology \cite{Gordon et al 2008}. Various methods have been proposed for decoherence control including quantum error-avoiding codes \cite{Zanardi and Rasetti 1997}-\cite{Kwiat et al 2000}, quantum error-correction codes \cite{Knill et al 2000}, dynamical decoupling \cite{Viola et al 1999}, \cite{Khodjasteh et al 2010} and quantum feedback control \cite{Vitali et al 1997}. In quantum error-avoiding codes, quantum information is encoded in a decoherence free subspace which is inherently immune to decoherence due to specific symmetries in the system-environment interaction \cite{Protopopescu et al 2003}. Quantum error-correction codes are active methods to detect and counteract the effects of errors during quantum information processing via encoding redundant qubits. Dynamical decoupling of decoherence control is an open-loop control approach which often employs bang-bang control pulses to dynamically cancel the effect of decoherence. Quantum feedback and optimal control theory also provide powerful tools for the analysis and design of decoherence control \cite{ZhangJ et al 2005}, \cite{Cui et al 2008}. However, there are few results which consider robustness when uncertainties or inaccurate parameters exist in the system Hamiltonian or the system-environment interaction. Here we consider a robust decoherence control scheme for quantum systems subject to Markovian decoherence \cite{Nielsen and Chuang 2000}. In particular, we will focus on a single qubit subject to amplitude damping decoherence, phase damping decoherence and depolarizing decoherence \cite{Nielsen and Chuang 2000}. We propose a sampling-based design approach to guarantee the robustness of a single qubit system with uncertainties in the system Hamiltonian and the coupling strength of the system-environment interaction.

The paper is organized as follows. Section \ref{Sec2} presents the control problem formulation and defines the required robustness. In Section \ref{Sec3}, we present the main methods and results for robust control design. Section \ref{Sec4} gives the proofs of the main
results. Concluding remarks are given in Section \ref{Sec5}.

\section{Control problem formulation}\label{Sec2}

For an open quantum system, its state is described by the positive Hermitian density matrix (or density operator) $\rho$ satisfying $\text{tr}\rho=1$, and the evolution of $\rho$ cannot generally
be described in terms of a unitary transformation. In many
situations, a quantum master equation for $\rho(t)$ (or $\rho_{t}$) is a suitable
way to describe the dynamics of an open quantum system. One of the
simplest cases is when a Markovian approximation can be applied
under the assumption of a short environmental correlation time permitting the neglect of memory
effects \cite{Breuer and Petruccione 2002}. For an
$N$-dimensional open quantum system with Markovian dynamics, its
state $\rho(t)$ can be described by the following Markovian master
equation (for details, see, e.g., \cite{Breuer and Petruccione
2002}, \cite{Lindblad 1976}, \cite{Alicki and Lendi 2007}):
\begin{equation}\label{MME}
\dot{\rho}(t)=-i[H(t),\rho(t)]+\frac{1}{2}\sum_{j,k=0}^{N^{2}-1}\alpha_{jk}
\{[L_{j}\rho(t),L^{\dagger}_{k}]+[L_{j},\rho(t)L^{\dagger}_{k}]\}.
\end{equation}
Here for an arbitrary operator $X$, $[X, \rho]=X\rho-\rho X$ is the commutation operator, $\{L_{j}\}_{j=0}^{N^{2}-1}$ is a basis for the space of linear
bounded operators on the Hilbert space $\mathcal{H}$ with $L_{0}=I$, the coefficient
matrix $A=(\alpha_{jk})$ is positive semidefinite and physically
specifies the relevant relaxation rates and we have set $\hbar=1$ in this paper. Markovian master equations
have been widely used to model controlled quantum systems in quantum control \cite{Ticozzi and Viola 2008}-\cite{Ticozzi and Viola 2009}, especially for Markovian quantum feedback \cite{Wiseman and Milburn 2009}.

In this paper, we will focus on a two-level quantum system (a single qubit) with Markovian dynamics whose evolution can be described by the following Lindblad equation:
\begin{equation}
\dot{\rho}(t)=-i[H(t),\rho(t)]+\sum_{k=1}^{K}\gamma_{k} \mathfrak{D}[L_{k}]\rho(t),
\end{equation}
where $$\mathfrak{D}[L_{k}]\rho=L_{k}\rho L_{k}^{\dag}-\frac{1}{2}L_{k}^{\dag}L_{k}\rho-\frac{1}{2}\rho L_{k}^{\dag} L_{k}.$$
For such a single qubit system, we can divide $H(t)$ into three parts $H(t)=H_{0}+H_{\Delta}+H_{u}$, where the free Hamiltonian is $H_{0}=\frac{1}{2}\sigma_{z}$, the
control Hamiltonian is $H_{u}=\sum_{j=x,y,z} u_{j}(t)I_{j}$,
($u_{j}(t)\in \mathbf{R}$, $I_{j}=\frac{1}{2}\sigma_{j}$), and the uncertainties in the system Hamiltonian are
$H_{\Delta}=\omega(t)I_{z}+\epsilon_{x}(t)I_{x}+\epsilon_{y}(t)I_{y}$
($\omega(t), \epsilon_{x}(t), \epsilon_{y}(t)\in \mathbf{R}$).
The Pauli matrices $\sigma=(\sigma_{x},\sigma_{y},\sigma_{z})$ take the
following form:
\begin{equation}
\sigma_{x}=\begin{pmatrix}
  0 & 1  \\
  1 & 0  \\
\end{pmatrix} , \ \ \ \
\sigma_{y}=\begin{pmatrix}
  0 & -i  \\
  i & 0  \\
\end{pmatrix} , \ \ \ \
\sigma_{z}=\begin{pmatrix}
  1 & 0  \\
  0 & -1  \\
\end{pmatrix} .
\end{equation}
$H_{\Delta}$ is the first class of uncertainties we will consider in this paper.
The unitary errors in \cite{Pravia et al 2003} belong to this class
of uncertainties, and one-qubit gate errors also correspond to this class of uncertainties
\cite{Dong and Petersen 2009NJP}.
A second class of uncertainties are uncertainties $\delta\gamma_{k}$ residing in the coupling strength $\gamma_{k}$. Since the Lindblad equation is an approximate equation for the open quantum system coupling with its environment, this class of uncertainties may come from inaccurate modeling as well as time-varying coupling between the system and environment. We assume that all the uncertainties are bounded, i.e., $|\omega(t)|\leq \omega$, $\sqrt{\epsilon_{x}^{2}(t)+\epsilon_{y}^{2}}\leq \epsilon$ and $|\delta\gamma_{k}|\leq \gamma$, where constants $\omega\geq 0$, $\epsilon>0$ and $\gamma \geq 0$ are given.

For a qubit system, its state $\rho$ can be represented in terms of the
Bloch vector
$\mathbf{r}=(x,y,z)=(\text{tr}\{\rho\sigma_{x}\},\text{tr}\{\rho\sigma_{y}\},\text{tr}\{\rho\sigma_{z}\})$:
\begin{equation}\label{blochEq}
\rho=\frac{I+\mathbf{r}\cdot \sigma}{2} .
\end{equation}
After representing the state $\rho$ with the
Bloch vector, the pure states (i.e., with $\text{tr}(\rho^{2})=1$) for the qubit system lie on the surface of
the Bloch sphere and the mixed states (i.e., with $\text{tr}(\rho^{2})< 1$) occupy the interior of the Bloch sphere. The purity of $\rho$ is defined as $P=\text{tr}(\rho^{2})$. A pure state can also be represented by a unit vector $|\psi\rangle$ in a complex Hilbert space, where $\rho=|\psi\rangle \langle\psi|$, $\langle \psi|=(|\psi\rangle)^{\dagger}$ and the operation $X^{\dagger}$ refers to the adjoint of $X$. The fidelity of an arbitrary state $\rho$ in terms of $|\psi\rangle$ can be defined as $\langle\psi|\rho|\psi\rangle$. Thus, the fidelity between two pure states $|\psi\rangle$ and $|\phi\rangle$ reduces to $|\langle\psi|\phi\rangle|^{2}$. A projective measurement with $\sigma_{z}$ on the qubit in state $\rho$ will make the state collapse into $|0\rangle$ with probability $\langle 0|\rho|0\rangle$ or into $|1\rangle$ with probability $\langle 1|\rho|1\rangle$ (such a process is referred as the measurement collapse postulate), where $|0\rangle$ and $|1\rangle$ are the eigenstates of $\sigma_{z}$ with corresponding eigenvalues 1 and -1, respectively. Another useful quantity is the coherence which can be defined as $C=x^{2}+y^{2}$, where $x=\text{tr}(\rho\sigma_{x})$ and $y=\text{tr}(\rho\sigma_{y})$ (see, e.g., \cite{Lidar and Schneider 2005}, \cite{ZhangM et al 2007}, \cite{Zhang et al 2010}). A decoherence process due to the interaction of a quantum system with its environment may reduce its purity or coherence.

We will consider the following four cases in this paper:

A) No decoherence (i.e., $\gamma_{k}\equiv 0$). This case corresponds to a closed quantum system with a pure state $\rho_{t}$ satisfying the Schr\"{o}dinger equation
\begin{equation}\label{Case1Eq}
\dot{\rho}_t=-i[H(t), \rho_t].
\end{equation}

B) Amplitude damping decoherence. In this case, the population of the quantum system can change (e.g., through loss of energy by spontaneous emission). The evolution of $\rho_{t}$ can be described by the following equation:
 \begin{equation}\label{AD_masterEq}
\dot{\rho_t}=-i[H(t), \rho_t]+\gamma_{t}(\sigma_{-}\rho_t \sigma_{+}-\frac{1}{2}\sigma_{+}\sigma_{-}\rho_t-\frac{1}{2}\rho_{t}\sigma_{+}\sigma_{-})
\end{equation}
where $\sigma_{-}=\frac{1}{2}(\sigma_{x}-i\sigma_{y})$, $\sigma_{+}=\frac{1}{2}(\sigma_{x}+i\sigma_{y})$ $\gamma_{t}=\gamma_{0}+\delta \gamma_{t}$ and $|\delta \gamma_{t}|\leq \gamma$. We also assume that $\gamma_{0}\geq \gamma$, which guarantees the coupling strength $\gamma_{t}\geq 0$.

C) Phase damping decoherence. In this case, a loss of quantum coherence can occur without loss of energy in the quantum system. The evolution of the state may be described by the following equation:
\begin{equation}\label{PD_masterEq}
\dot{\rho_{t}}=-i[H(t), \rho_{t}]+\gamma_{t}(\sigma_{z}\rho_{t} \sigma_{z}-\rho_{t}).
\end{equation}

D) Depolarizing decoherence. This decoherence maps pure states into mixed states. The dynamics can be described by the following equation:
\begin{equation}\label{DD_masterEq}
\dot{\rho_{t}}=-i[H(t), \rho_{t}]+\gamma_{t}(\sigma_{x}\rho_{t} \sigma_{x}-\rho_{t})+\gamma_{t}(\sigma_{y}\rho_{t} \sigma_{y}-\rho_{t})+\gamma_{t}(\sigma_{z}\rho_{t} \sigma_{z}-\rho_{t}).
\end{equation}

The objective of this paper is to design control laws for single qubit systems guaranteeing required robustness with the two classes of uncertainties. The required robustness for the four cases above is defined using the concept of a sliding mode domain, respectively, as follows.

\begin{definition}\cite{Dong and Petersen 2011Automatica}\label{definition1}
The sliding mode domain for a single qubit system without decoherence (closed systems) is defined as $\mathcal{D}_{c}=\{|\psi\rangle :  |\langle 0|\psi\rangle|^{2}\geq
1-p_{0}, 0< p_{0}< 1\}$.
\end{definition}

\begin{definition}\label{definition2}
The sliding mode domain for an open qubit system with amplitude damping decoherence is defined as $\mathcal{D}_{a}=\{\rho :  \langle 0|\rho|0\rangle\geq
1-p_{0}, 0< p_{0}< 1\}$.
\end{definition}

\begin{definition}\label{definition3}
The sliding mode domain for an open qubit system with phase damping decoherence is defined as $\mathcal{D}_{p}=\{\rho :  x^{2}+y^{2}\geq
\bar{C}, x=\text{tr}(\rho\sigma_{x}), y=\text{tr}(\rho\sigma_{y}), 0< \bar{C}\leq 1\}$.
\end{definition}

\begin{definition}\label{definition4}
The sliding mode domain for an open qubit system with depolarizing decoherence is defined as $\mathcal{D}_{d}=\{\rho :  \text{tr}\rho^{2} \geq
\bar{P}, 0.5< \bar{P}\leq 1\}$.
\end{definition}

\begin{remark}
The definition of $\mathcal{D}_{c}$ implies that the system's
state has a probability of at most $p_{0}$ (which we call the
probability of failure) to collapse out of $\mathcal{D}_{c}$ when making
a projective measurement with the operator $\sigma_{z}$. We aim
to drive and then maintain a single qubit's state in the
sliding mode domain $\mathcal{D}_{c}$. However, the uncertainties
$H_{\Delta}$ may take the system's state away from $\mathcal{D}_{c}$.
The sampling process (a measurement operation) unavoidably makes
the sampled system's state change. Thus, we expect that the control law
will guarantee that the system's state remains in $\mathcal{D}_{c}$,
except that the sampling process may take it away from $\mathcal{D}_{c}$
with a small probability (not greater than $p_{0}$). The definition of $\mathcal{D}_{a}$
has a similar meaning to $\mathcal{D}_{c}$. The difference lies in the fact that the quantum state
in Definition \ref{definition2} could be a mixed state $\rho$ and the system is also subject to amplitude damping decoherence. From Definition \ref{definition3}, we know
all states in $\mathcal{D}_{p}$ have coherence of at least $\bar{C}$. Definition \ref{definition4} defines $\mathcal{D}_{d}$ as a set where the purity of an arbitrary
quantum state is not less than $\bar{P}$.
\end{remark}

\section{Main Methods and Results}\label{Sec3}
In this paper, we propose a sampled-data design method for robust control of quantum systems with uncertainties. A key task is to design a sampling period as large as possible to guarantee the required robustness defined using a sliding mode domain. The sampling process is taken as an important control tool to modify the system dynamics nonunitarily. For the cases without decoherence and with amplitude damping decoherence, it is also necessary to design a control law to drive the system's state back to the corresponding sliding mode domain when the sampling process makes the system's state collapse out of the sliding mode domain. Such a control law corresponds to a unitary transformation and we refer to it as ``unitary control" in this paper. The sequel will provide the main methods and results for the four cases of uncertain quantum systems and then present some illustrative examples.

\subsection{No decoherence}
The objective is to develop a control strategy to guarantee the  required robustness when
bounded uncertainties exist in the system Hamiltonian. According to Definition 1, we specify the required
robustness as follows: (a) maintain the system's state in
the sliding mode domain $\mathcal{D}_{c}$ in which the system's state has
a high fidelity ($\geq 1-p_{0}$) with the sliding mode state $|0\rangle$, and (b)
once the system's state collapses out of $\mathcal{D}_{c}$ upon making a
measurement (sampling), drive it back to $\mathcal{D}_{c}$ within a
short time period $\beta T_{c}$ and maintain the state in
$\mathcal{D}_{c}$ for the following time period $(1-\beta) T_{c}$ (where $0\leq
\beta< 1$ and $T_{c}$ is the sampling period). $\beta$ is used to characterize the proportion of time that the unitary control is applied within the corresponding sampling period. Generally we choose
$\beta$ to satisfy $\frac{\beta}{1-\beta}\ll 1$, and this assumption will be helpful for designing the unitary control, which is demonstrated in the examples. To guarantee the
required robustness, we design a control law based on sampled-data
measurements as follows: For any sampling time $nT_{c}$ ($n=0, 1, 2, \dots$), (i) if
the measurement data corresponds to $|0\rangle$, let the system evolve
with zero control and sample again at the time $(n+1)T_{c}$; (ii)
otherwise, apply a unitary control to drive
the system's state back into a subset $\mathcal{E}_{c}$ of $\mathcal{D}_{c}$ from the time
$nT_{c}$ to $(n+\beta)T_{c}$, then sample again at time $(n+1)T_{c}$. The
control operation is switched between (i) and (ii) based on the
sampled data measurements. In (ii), to guarantee the desired goal
when $t\in [(n+1-\beta)T_{c}, (n+1)T_{c}]$, the unitary control should drive the system's state into $\mathcal{E}_{c}\subset\mathcal{D}_{c}$. $\mathcal{E}_{c}$ can be defined as
$\mathcal{E}_{c}=\{|\psi\rangle :  |\langle 0|\psi\rangle|^{2}\geq
1-\alpha p_{0}, 0< p_{0}< 1, 0\leq \alpha\leq 1 \}$.  The sampling period $T_{c}$ and the unitary control can
be designed offline in advance. The basic method we use is illustrated in Fig. \ref{scheme1}. The sequel will outline the design of the sampling period and establish a relationship between $\alpha$ and $\beta$ to guarantee the required robustness.

\begin{figure}
\centering
\includegraphics[width=5.6in]{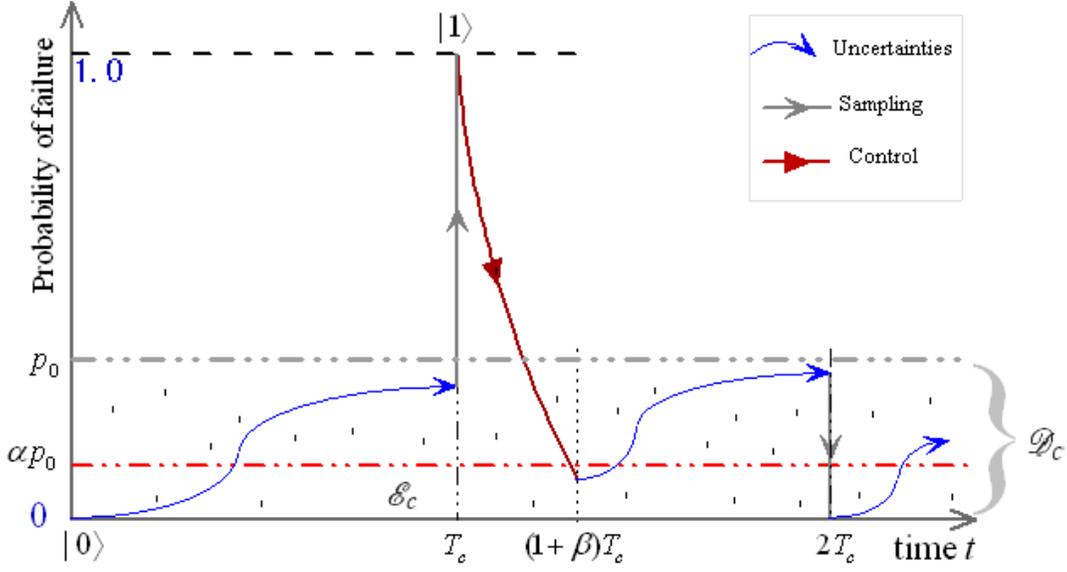}
\caption{The proposed sampled-data control scheme for a single qubit
system without decoherence. The labels ``Control",
``Sampling" and ``Uncertainties" refer to  the evolution process under the unitary control, the sampling process,
and uncertainties in the system Hamiltonian, respectively.}
\label{scheme1}
\end{figure}

Using a similar argument to Theorem 1 in \cite{Dong and Petersen
2011Automatica}, we have the following result \cite{Dong and Petersen 2011IFAC}.
\begin{lemma}\label{lemma1}
For a single qubit with initial state
$|\psi(0)\rangle=|0\rangle$ at time $t=0$, the system evolves to
$|\psi(t)\rangle$ under the action of
$H(t)=[1+\omega(t)]I_{z}+\epsilon_{x}(t) I_{x}+\epsilon_{y}(t)
I_{y}$ (where $|\omega(t)|\leq \omega$, $\omega\geq 0$,
$\sqrt{\epsilon_{x}^{2}(t)+\epsilon_{y}^{2}(t)}\leq \epsilon$ and
$\epsilon >0$). If $t\in [0, T_{c}]$, where
\begin{equation}\label{period}
T_{c}=\frac{\arccos(1-2p_{0})}{\epsilon},
\end{equation}
the state will remain in $\mathcal{D}_{c}=\{|\psi\rangle :
|\langle 0 |\psi\rangle|^{2}\geq 1-p_{0}\}$ (where $0< p_{0}< 1$).
When a projective measurement is made with the
operator $\sigma_{z}$ at time $t$, the probability of failure
$p=|\langle 1|\psi(t)\rangle|^{2}$ is not greater than $p_{0}$.
\end{lemma}

%For imperfect measurement, we have the following result.
%\begin{theorem}\label{CD_Theorem1}
%For a single qubit system with the initial state
%$|\psi(0)\rangle=|0\rangle$ at the time $t=0$, the system evolves to
%$|\psi(t)\rangle$ under the action of
%$H(t)=[1+\omega(t)]I_{z}+\epsilon_{x}(t) I_{x}+\epsilon_{y}(t)
%I_{y}$ (where $|\omega(t)|\leq \omega$, $0\leq \omega$,
%$\sqrt{\epsilon_{x}^{2}(t)+\epsilon_{y}^{2}(t)}\leq \epsilon$ and
%$\epsilon >0$). If $t\in [0, T_{c}]$, where
%\begin{equation}\label{period_imperfect}
%T_{c}=\frac{\arccos(1-\frac{2(p_{0}-p_{01})}{1-p_{10}})}{\epsilon},
%\end{equation}
%the system's state will remain in $\mathcal{D}_{c}=\{|\psi\rangle :
%|\langle 0 |\psi\rangle|^{2}\geq 1-p_{0}'\}$ (where $0< p_{0}'< 1$).
%When one makes a projective measurement with the imperfect measurement model
%at the time $t$, the probability of failure
%$p=|\langle 1|\psi(t)\rangle|^{2}$ is not greater than $p_{0}$.
%\end{theorem}

We use $T_{c}$ defined in (\ref{period}) as the sampling period
to guarantee the required performance. If the sampling data
corresponds to $|1\rangle$, a unitary control is required to drive the
state back to a subset $\mathcal{E}_{c}$ of $\mathcal{D}_{c}$. The following theorem
gives a sufficient condition on the relationships between $\alpha$,
$p_{0}$ and $\beta$ to guarantee the required robustness.
%In this paper, we have assumed that the possible uncertainties can
%be described by
%$H_{\Delta}=\epsilon_{x}(t)I_{x}+\epsilon_{y}(t)I_{y}+\delta(t)I_{z}$,
%where unknown $\epsilon_{x}(t)$, $\epsilon_{y}(t)$ and $\delta(t)$
%satisfy $\sqrt{\epsilon_{x}^{2}(t)+\epsilon_{y}^{2}(t)}\leq
%\epsilon$ and $|\delta(t)|\leq \delta$. We have the following
%theorem.
\begin{theorem}\label{CD_Theorem2}
For a single qubit with initial state
$|\psi(0)\rangle$ satisfying $|\langle \psi(0)|1\rangle|^{2}$ $\leq
\alpha p_{0}$ ($0\leq \alpha \leq 1$) at time $t=0$, the system
evolves to $|\psi(t)\rangle$ under the action of
$H(t)=[1+\omega(t)]I_{z}+\epsilon_{x}(t) I_{x}+\epsilon_{y}(t)
I_{y}$ (where $\sqrt{\epsilon_{x}^{2}(t)+\epsilon_{y}^{2}(t)}\leq
\epsilon$, $\epsilon >0$, $|\omega(t)|\leq \omega$ and $\omega\geq 0$). If $t\in [0, (1-\beta) T_{c}]$ and
\begin{equation}\label{unitary1}
\alpha \leq
\frac{1-\cos[\beta\arccos(1-2p_{0})]}{2p_{0}}
\end{equation}
 where $0\leq\beta < 1$ and
\begin{equation}\label{2period}
T_{c}=\frac{\arccos(1-2p_{0})}{\epsilon},
\end{equation}
then the state will remain in $\mathcal{D}_{c}=\{|\psi\rangle :
|\langle 0 |\psi\rangle|^{2}\geq 1-p_{0}\}$ (where $0< p_{0}<1$).
When a projective measurement is made with the
operator $\sigma_{z}$ at time $t$, the probability of failure
$p=|\langle 1|\psi(t)\rangle|^{2}$ is not greater than $p_{0}$.
\end{theorem}

\begin{remark}
Using Lemma \ref{lemma1} and Theorem \ref{CD_Theorem2}, we
aim to maintain the state in $\mathcal{D}_{c}$ by
implementing periodic sampling with period $T_{c}$ in
(\ref{period}). This theorem provides a sufficient condition to
guarantee the required robustness. Given $p_{0}$, $\beta$, we can
select $\alpha$ satisfying (\ref{unitary1}) in Theorem
\ref{CD_Theorem2}. If the sampled result is $|1\rangle$, we apply a
unitary control to drive the state into $\mathcal{E}_{c}$.
The sampling period and the unitary control can be designed
in advance. Different approaches can be used to design such a unitary control law. In this paper, we will employ a Lyapunov method \cite{Mirrahimi et al 2005}-\cite{Kuang and Cong 2008} in Example 2 to accomplish this task for the closed quantum system.
\end{remark}

\begin{remark}
The design
scheme above involves a sampling process and a unitary control. It is similar to the approach used in \cite{Dong and Petersen 2011Automatica}. The difference lies in the fact that the scheme in this paper involves a fixed sampling period $T_{c}$. However, the approach in \cite{Dong and Petersen 2011Automatica} involves at
least two measurement periods $T$ (equivalent to $T_{c}$ in this paper) and $T_{1}$ ($T_{1}\ll T$). This situation
means that the approach of \cite{Dong and Petersen 2011Automatica}
may require measurements which are very close together, which may
be difficult to achieve in practice. In this sense, the sampled-data
design in this paper is more practical than the
method in \cite{Dong and Petersen 2011Automatica}.
\end{remark}

\subsection{Amplitude damping decoherence}
For single qubit systems with amplitude damping decoherence, if the initial state is excited state $|0\rangle$, the decoherence will drive this excited state to the ground state $|1\rangle$. The objective is to design a control law to guarantee the required robustness defined by $\mathcal{D}_{a}$. We use a similar sampled-data design method to that in the case without decoherence. That is, if the state is $|0\rangle$ at $t=nT_{a}$ ($n=0, 1, 2, \dots$), we design a sampling period to maintain the system's state in $\mathcal{D}_{a}$ by implementing periodic sampling with period $T_{a}$; if the measurement makes the state collapse into $|1\rangle$ (with a probability $p\leq p_{0}$), we design a unitary control to drive the state back into a subset $\mathcal{E}_{a}$ of $\mathcal{D}_{a}$ from $t=nT_{a}$ to $(n+\beta)T_{a}$, and then sample again at $t=(n+1)T_{a}$. In order to determine the required sampling period, we have the following results.

\begin{theorem}\label{AC_Thereom1}
For a single qubit with initial state
$|0\rangle$ at time $t=0$, the system
evolves to $\rho(t)$ subject to (\ref{AD_masterEq}) where
$H(t)=[1+\omega(t)]I_{z}+\epsilon_{x}(t) I_{x}+\epsilon_{y}(t)
I_{y}$ ($\sqrt{\epsilon_{x}^{2}(t)+\epsilon_{y}^{2}(t)}\leq
\epsilon$, $\epsilon >0$, $|\omega(t)|\leq \omega$ and $\omega\geq 0$) and the coupling strength of amplitude damping decoherence is $\gamma_{t}=\gamma_{0}+\delta\gamma_{t}$ ($|\delta\gamma_{t}|\leq \gamma$). If $t\in [0, T_a]$ with
\begin{equation}\label{2period}
T_a=\frac{2p_{0}}{\sqrt{4\epsilon^{2}+(\gamma_{0}+\gamma)^{2}}+(\gamma_{0}+\gamma)},
\end{equation}
the state will remain in $\mathcal{D}_{a}=\{\rho :
\langle 0 |\rho|0\rangle\geq 1-p_{0}\}$ (where $0< p_{0}<1$).
When a projective measurement is made with the
operator $\sigma_{z}$ at time $t$, the probability of failure
$p=\langle 1|\rho|1\rangle$ is not greater than $p_{0}$.
\end{theorem}

\begin{corollary}\label{AC_improved_proposition1}
%For a single qubit system with the initial state
%$|0\rangle$ at the time $t=0$, the system
%evolves to $\rho(t)$ subjected to (\ref{AD_masterEq}) where
%$H(t)=[1+\omega(t)]I_{z}+\epsilon_{x}(t) I_{x}+\epsilon_{y}(t)
%I_{y}$ ($\sqrt{\epsilon_{x}^{2}(t)+\epsilon_{y}^{2}(t)}\leq
%\epsilon$, $\epsilon >0$, $|\omega(t)|\leq \omega$ and $\omega\geq 0$) and the coupling strength of amplitude damping %decoherence $\gamma_{t}=\gamma_{0}+\delta\gamma_{t}$ ($|\delta\gamma_{t}|\leq \gamma$).
If $p_{0}\leq \frac{1}{2}-\frac{\gamma_{0}+\gamma}{2\sqrt{4\epsilon^{2}+(\gamma_{0}+\gamma)^{2}}}$ and $t\in [0, T_a']$, the sampling period $T_{a}$ in (\ref{2period}) can be replaced by $T_{a}'$ to guarantee the same robustness as in Theorem \ref{AC_Thereom1}, where
\begin{equation}\label{3period}
T_{a}'=\frac{2p_{0}}{4\epsilon\sqrt{p_{0}-p_{0}^{2}}+2(\gamma_{0}+\gamma)(1-p_{0})}.
\end{equation}
%the system's state will remain in $\mathcal{D}_{a}=\{\rho :
%\langle 0 |\rho|0\rangle\geq 1-p_{0}\}$ (where $0< p_{0}<1$).
%When one makes a projective measurement with the measurement
%operator $\sigma_{z}$ at the time $t$, the probability of failure
%$p=\langle 1|\rho|1\rangle$ is not greater than $p_{0}$.
\end{corollary}

When there exist no uncertainties in the system Hamiltonian (i.e., $H_{\Delta}\equiv 0$), the sampling period can be designed using the following proposition.
\begin{proposition}\label{AC_proposition2}
For a single qubit with initial state
$|0\rangle$ at time $t=0$, the system
evolves to $\rho(t)$ subject to (\ref{AD_masterEq}) where
$H(t)=I_{z}$ and the coupling strength of amplitude damping decoherence is $\gamma_{t}=\gamma_{0}+\delta\gamma_{t}$ ($|\delta\gamma_{t}|\leq \gamma$). If $t\in [0, T_a'']$ with
\begin{equation}\label{4period}
T_a''=-\frac{\ln (1-p_{0})}{\gamma_{0}+\gamma},
\end{equation}
the state will remain in $\mathcal{D}_{a}=\{\rho :
\langle 0 |\rho|0\rangle\geq 1-p_{0}\}$ (where $0< p_{0}<1$).
When a projective measurement is made with the
operator $\sigma_{z}$ at time $t$, the probability of failure
$p=\langle 1|\rho|1\rangle$ is not greater than $p_{0}$.
\end{proposition}

\begin{remark}
From the proof of Proposition \ref{AC_proposition2}, it is clear that $T_a''=-\frac{\ln (1-p_{0})}{\gamma_{0}+\gamma}$ exactly corresponds to the case $\delta\gamma_{t}\equiv \gamma$ when $H_{\Delta}\equiv 0$. In this sense, the sampling period $T_a''=-\frac{\ln (1-p_{0})}{\gamma_{0}+\gamma}$ is optimal to guarantee the required robustness. From the proofs of Theorem \ref{AC_Thereom1} and Corollary \ref{AC_improved_proposition1}, it is clear that $T_{a}'\geq T_{a}$. The relationship $T_{a}''\geq T_{a}$ for arbitrary $p_{0}$ can be proved by the following steps: (a) Define $F(p_{0})=T_{a}''-T_{a}$; (b) observe $F(p_{0}=0)=0$; and (c) verify $\frac{\text{d}F}{\text{d}p_{0}}\geq 0$.
\end{remark}

Hence, for different situations we may use $T_{a}$, $T_{a}'$ or $T_{a}''$ as the sampling period
to guarantee the required performance. If the sampled data
corresponds to $|1\rangle$, a unitary control is required to drive the
state back to a subset $\mathcal{E}_{a}$ of $\mathcal{D}_{a}$. The subset $\mathcal{E}_{a}$ may be defined as
$\mathcal{E}_{a}=\{\rho :  \langle 0|\rho|0 \rangle\geq
1-\alpha p_{0}, 0< p_{0}< 1, 0\leq \alpha\leq 1 \}$. The following theorem
gives a sufficient condition on the relationships between $\alpha$,
$p_{0}$ and $\beta$ to guarantee the required robustness (The following conclusion is also true when $T_{a}$ can be replaced by $T_{a}'$ or $T_{a}''$).
%In this paper, we have assumed that the possible uncertainties can
%be described by
%$H_{\Delta}=\epsilon_{x}(t)I_{x}+\epsilon_{y}(t)I_{y}+\delta(t)I_{z}$,
%where unknown $\epsilon_{x}(t)$, $\epsilon_{y}(t)$ and $\delta(t)$
%satisfy $\sqrt{\epsilon_{x}^{2}(t)+\epsilon_{y}^{2}(t)}\leq
%\epsilon$ and $|\delta(t)|\leq \delta$. We have the following
%theorem.
\begin{theorem}\label{AD_Theorem2}
For a single qubit with initial state
$\rho(0)$ satisfying $\langle 1|\rho(0)|1\rangle$ $\leq
\alpha p_{0}$ ($0\leq \alpha \leq 1$) at time $t=0$, the system
evolves to $\rho(t)$ subject to (\ref{AD_masterEq}) where
$H(t)=[1+\omega(t)]I_{z}+\epsilon_{x}(t) I_{x}+\epsilon_{y}(t)
I_{y}$ ($\sqrt{\epsilon_{x}^{2}(t)+\epsilon_{y}^{2}(t)}\leq
\epsilon$, $\epsilon >0$, $|\omega(t)|\leq \omega$ and $\omega\geq 0$) and the coupling strength of amplitude damping decoherence is $\gamma_{t}=\gamma_{0}+\delta\gamma_{t}$ ($|\delta\gamma_{t}|\leq \gamma$). If $t\in [0, (1-\beta) T_{a}]$ and
\begin{equation}\label{sufficientCondition}
\alpha \leq
\beta
\end{equation}
where $0<\beta \leq 1$ and
\begin{equation}\label{AD2period}
T_{a}=\frac{2p_{0}}{\sqrt{4\epsilon^{2}+(\gamma_{0}+\gamma)^{2}}+(\gamma_{0}+\gamma)},
\end{equation}
the state will remain in $\mathcal{D}_{a}=\{\rho :
\langle 0|\rho|0\rangle\geq 1-p_{0}\}$ (where $0< p_{0}<1$).
When a projective measurement is made with the
operator $\sigma_{z}$ at time $t$, the probability of failure
$p=\langle 1|\rho|1\rangle$ is not greater than $p_{0}$.
\end{theorem}

%\begin{remark}
%If no uncertainties exist in the system Hamiltonian, the relationship can be improved as $\alpha\leq %\frac{(1-p_{0})^{1-\beta}+p_{0}-1}{p_{0}}$ with $T_{a}''$.
%\end{remark}

\subsection{Phase damping decoherence}
For a single qubit with phase damping decoherence, we define the coherence as $C=x^{2}+y^{2}$ where $x=\text{tr}(\rho\sigma_{x})$ and $y=\text{tr}(\rho\sigma_{y})$. The phase damping decoherence will reduce the coherence of the system. The objective is to guarantee that the state has coherence not less than $\bar{C}$ by periodic sampling when there exist uncertainties in the coupling strength of system-environment interaction and in the system Hamiltonian. To determine the required sampling period, we have the following results.
\begin{theorem}\label{PD_Theorem1}
For a single qubit with initial state
$\rho_{0}$ satisfying $C_{0}=[\text{tr}(\rho_{0}\sigma_{x})]^{2}+[\text{tr}(\rho_{0}\sigma_{y})]^{2}=1$ at time $t=0$, the system
evolves to $\rho_{t}$ subject to (\ref{PD_masterEq}) where
$H(t)=[1+\omega(t)]I_{z}+\epsilon_{x}(t) I_{x}+\epsilon_{y}(t)
I_{y}$ ($\sqrt{\epsilon_{x}^{2}(t)+\epsilon_{y}^{2}(t)}\leq
\epsilon$, $\epsilon >0$, $|\omega(t)|\leq \omega$ and $\omega\geq 0$) and the coupling strength of the phase damping decoherence is $\gamma_{t}=\gamma_{0}+\delta\gamma_{t}$ ($|\delta\gamma_{t}|\leq \gamma$). If $t\in [0, T_{p}]$ with
\begin{equation}\label{TP}
T_{p}=\begin{cases}\frac{1-\bar{C}}{4\sqrt{2}(\gamma_{0}+\gamma)}, \  \ \ \text{when} \ 4(\gamma_{0}+\gamma)^{2}\geq \epsilon^{2};\\
\frac{(1-\bar{C})\sqrt{\epsilon^{2}-2(\gamma_{0}+\gamma)^{2}}}{2\epsilon^{2}}, \  \ \ \text{when} \ 4(\gamma_{0}+\gamma)^{2}< \epsilon^{2},
\end{cases}
\end{equation}
the state will remain in $\mathcal{D}_{p}=\{\rho_{t} :
[\text{tr}(\rho_{t}\sigma_{x})]^{2}+[\text{tr}(\rho_{t}\sigma_{y})]^{2}\geq \bar{C}, 0< \bar{C}\leq1 \}$.
When a periodic projective measurement is made with the
operator $\sigma_{x}$ on the system, the sampling (measurement) period $T_{p}$ can guarantee that the state remains in $\mathcal{D}_{p}$.
\end{theorem}

\begin{corollary}\label{PD_proposition1}
%For a single qubit system with the initial state
%$\rho_{0}$ satisfying $x_{0}^{2}+y_{0}^{2}=1$ at the time $t=0$, the system
%evolves to $\rho(t)$ subjected to (\ref{PD_masterEq}) where
%$H(t)=[1+\omega(t)]I_{z}+\epsilon_{x}(t) I_{x}+\epsilon_{y}(t)
%I_{y}$ ($\sqrt{\epsilon_{x}^{2}(t)+\epsilon_{y}^{2}(t)}\leq
%\epsilon$, $\epsilon >0$, $|\omega(t)|\leq \omega$ and $\omega\geq 0$) and the coupling strength of phase damping %decoherence $\gamma_{t}=\gamma_{0}+\delta\gamma_{t}$ ($|\delta\gamma_{t}|\leq \gamma$).
When $\epsilon^{2}=2(\gamma_{0}+\gamma)^{2}$, the sampling period $T_{p}$ in (\ref{TP}) can be replaced by $T_{p}'$ to guarantee the same robustness as in Theorem \ref{PD_Theorem1}, where
\begin{equation}
T_{p}'=\frac{1-\sqrt{\bar{C}}}{2\sqrt{2}(\gamma_{0}+\gamma)}.
\end{equation}
%the system's state will remain in $\mathcal{D}_{p}=\{\rho :
%x_{t}^{2}+y_{t}^{2}\geq \bar{C}, x=\text{tr}(\rho\sigma_{x}), y=\text{tr}(\rho\sigma_{y}), 0< \bar{C}\leq 1 \}$.
%One makes periodic projective measurements with the measurement
%operator $\sigma_{x}$ and the sampling period $T_{p}'$ can guarantee the system's state in $\mathcal{D}_{p}$.
\end{corollary}

If $H_{\Delta}\equiv 0$, we can design the sampling period using the following proposition.
\begin{proposition}\label{PD_proposition2}
For a single qubit with initial state
$\rho_{0}$ satisfying $C_{0}=[\text{tr}(\rho_{0}\sigma_{x})]^{2}+[\text{tr}(\rho_{0}\sigma_{y})]^{2}=1$ at time $t=0$, the system
evolves to $\rho(t)$ subject to (\ref{PD_masterEq}) where
$H(t)=I_{z}$ and the coupling strength of the phase damping decoherence is $\gamma_{t}=\gamma_{0}+\delta\gamma_{t}$ ($|\delta\gamma_{t}|\leq \gamma$). If $t\in [0, T_{p}'']$ with
\begin{equation}
T_{p}''=-\frac{\ln \bar{C}}{4(\gamma_{0}+\gamma)},
\end{equation}
the state will remain in $\mathcal{D}_{p}=\{\rho_{t} :
[\text{tr}(\rho_{t}\sigma_{x})]^{2}+[\text{tr}(\rho_{t}\sigma_{y})]^{2}\geq \bar{C}, 0< \bar{C}\leq1 \}$.
If a periodic projective measurement is made with the
operator $\sigma_{x}$, the sampling period $T_{p}''$ can guarantee that the system's state remains in $\mathcal{D}_{p}$.
\end{proposition}

\begin{remark}
For $2(\gamma+\gamma_{0})^{2}=\epsilon^{2}$,
it is straightforward to prove that $T_{p}'\geq T_{p}$. The relationship $T_{p}''\geq T_{p}$ can be proved by the following steps: (a) Denote $Y=1-\bar{C}$; (b) define $F(Y)=T_{p}''- T_{p}$; (c) observe $F(Y=0)=0$; and (d) verify $\frac{\text{d}F(Y)}{\text{d}Y}\geq 0$. From the proof of Proposition \ref{PD_proposition2}, it is clear that the sampling period $T_{p}''$ is optimal to guarantee the required robustness when $H_{\Delta}\equiv 0$. For this case with phase damping decoherence, we can also make projective measurements with the operator $\sigma_{y}$, which does not affect the conclusions. Moreover, in this case, no unitary control is required and measurement is the only tool needed for guaranteeing the required robustness. It is worth mentioning that several methods based only on  measurements have recently been proposed for controlling quantum systems (see, e.g., \cite{Roa et al 2006}-\cite{Ashhab and Nori 2010}).
\end{remark}

\begin{remark}
We can also consider a class of imperfect measurements. This class of uncertainties may arise from precision limitations of the measurement apparatus or from system errors in the measurement device.
Measurement with the operator $\sigma_{z}$ will make the system collapse into $|0\rangle$ or $|1\rangle$ (eigenstates of $H_{0}$). We consider the imperfect measurement model shown as in Fig. \ref{MeasurementModel}. $p_{01}$ is the error probability of measurement from $|0\rangle$ to $|1\rangle$, that is, the probability that one obtains the result $|1\rangle$ when making a measurement on the system in $|0\rangle$; $p_{10}$ is the error probability of measurement from $|1\rangle$ to $|0\rangle$, where $0\leq p_{01}<1$ and $0\leq p_{10}<1$. This class of imperfect measurements does not affect the effectiveness of the sampled-data design. Thus, the proposed method can tolerate this additional uncertainty in the sampling process.
%For example, when we use a photon detector to detect 1000 photons from a system, we might only read out 999 photons.

\begin{figure}
\centering
\includegraphics[width=3.2in]{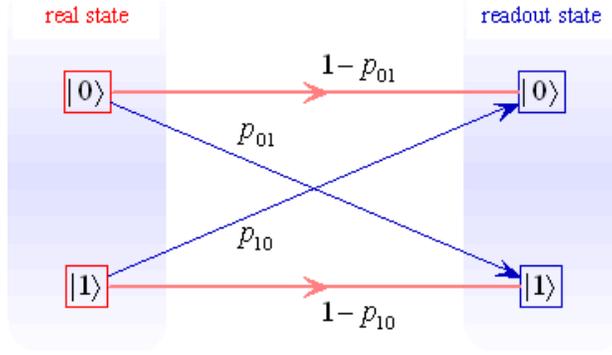}
\caption{The model for imperfect measurement. $p_{01}$ is the error probability of measurement from $|0\rangle$ to $|1\rangle$ and $p_{10}$ is the error probability of measurement from $|1\rangle$ to $|0\rangle$.}
\label{MeasurementModel}
\end{figure}
\end{remark}

\subsection{Depolarizing decoherence}
For a single qubit, the depolarizing decoherence will reduce the purity $P=\text{tr}(\rho^{2})$ of the system's state. The objective is to guarantee that the purity of the state is not less than $\bar{P}$ by periodic sampling when there exist uncertainties in the coupling strength of system-environment interaction and in the system Hamiltonian. To determine the required sampling period, we have the following results.

\begin{theorem}\label{DD_Theorem}
For a single qubit with initial state
$\rho_{0}$ satisfying $\text{tr}(\rho_{0}^{2})=1$ at time $t=0$, the system
evolves to $\rho_{t}$ subject to (\ref{DD_masterEq}) where
$H(t)=[1+\omega(t)]I_{z}+\epsilon_{x}(t) I_{x}+\epsilon_{y}(t)
I_{y}$ ($\sqrt{\epsilon_{x}^{2}(t)+\epsilon_{y}^{2}(t)}\leq
\epsilon$, $\epsilon >0$, $|\omega(t)|\leq \omega$ and $\omega\geq 0$) and the coupling strength of depolarizing decoherence is $\gamma_{t}=\gamma_{0}+\delta\gamma_{t}$ ($|\delta\gamma_{t}|\leq \gamma$). If $t\in [0, T_{d}]$ with
\begin{equation}
T_{d}=-\frac{\ln (2\bar{P}-1)}{8(\gamma_{0}+\gamma)},
\end{equation}
the state will remain in $\mathcal{D}_{d}=\{\rho_{t} :
\text{tr}(\rho_{t}^{2})\geq \bar{P}, 0.5< \bar{P}\leq 1\}$.
If periodic projective measurements are made with the
operator $\sigma_{z}$, the sampling period $T_{d}$ can guarantee that the state remains in $\mathcal{D}_{d}$.
\end{theorem}

\begin{remark}
The selection of measurement operators (i.e., $\sigma_{x}$, $\sigma_{y}$ or $\sigma_{z}$), uncertainties in the system Hamiltonian ($H_{\Delta}\neq 0$ or $H_{\Delta}\equiv 0$) and the imperfect measurement described in Fig. \ref{MeasurementModel} do not affect the conclusion in Theorem \ref{DD_Theorem}. The sampling period $T_{d}$ is also optimal to guarantee the required robustness.
\end{remark}

The sampling periods for the different cases considered above are summarized in Table 1.

\newpage

Table 1: Summary of sampling periods for different cases. $H_{\Delta}(x,y,z)=\omega(t)I_{z}+\epsilon_{x}(t) I_{x}+\epsilon_{y}(t)I_{y}$ (where $\sqrt{\epsilon_{x}^{2}(t)+\epsilon_{y}^{2}(t)}\leq
\epsilon$, $\epsilon >0$, $|\omega(t)|\leq \omega$ and $\omega\geq 0$), $f(\epsilon, \gamma_{0},\gamma)= \frac{1}{2}-\frac{\gamma_{0}+\gamma}{2\sqrt{4\epsilon^{2}+(\gamma_{0}+\gamma)^{2}}}$, and coupling strength $\gamma_{t}=\gamma_{0}+\delta\gamma_{t}$ ($|\delta\gamma_{t}|\leq \gamma$). $H_{\Delta}(x,y,z)$ is also considered for the two cases $p_{0}\leq f(\epsilon, \gamma_{0},\gamma)$ and $\epsilon^{2}=2(\gamma_{0}+\gamma)^{2}$. The parameters values $p_{0}=0.01$, $\gamma_{0}=0.9$, $\gamma=0.1$ and $\bar{C}=\bar{P}=0.95$ are assumed for the calculation of the right two columns. When $4(\gamma_{0}+\gamma)^{2}\geq \epsilon^{2}$, $\overline{T}=\frac{1-\bar{C}}{4\sqrt{2}(\gamma_{0}+\gamma)}$;
when $4(\gamma_{0}+\gamma)^{2}< \epsilon^{2}$, $\overline{T}=\frac{(1-\bar{C})\sqrt{\epsilon^{2}-2(\gamma_{0}+\gamma)^{2}}}{2\epsilon^{2}}$.

\begin{table}[!htbp]
\centering
\begin{tabular}{|c|c|c|c|c|}
\hline
\multicolumn{2}{|c|}{cases} &  sampling period & $\epsilon=0.2$ & $\epsilon=\sqrt{2}$ \\
\hline
\multicolumn{2}{|c|}{closed system with $H_{\Delta}(x,y,z)$} & $T_{c}=\frac{\arccos(1-2 p_{0})}{\epsilon}$ & $T_{c}=1.0017$  & $T_{c}=0.1417$ \\
\hline
%\multirow{3}{*}
amplitude & $H_{\Delta}(x,y,z)$ & $T_a=\frac{2p_{0}}{\sqrt{4\epsilon^{2}+(\gamma_{0}+\gamma)^{2}}+(\gamma_{0}+\gamma)}$ & $T_{a}=0.0096$ & $T_{a}=0.0050$  \\
\cline{2-5}
damping & $p_{0}\leq f(\epsilon, \gamma_{0},\gamma)$ & $T_{a}'=\frac{2p_{0}}{4\epsilon\sqrt{p_{0}-p_{0}^{2}}+2(\gamma_{0}+\gamma)(1-p_{0})}$ & $-$ & $T_{a}'=0.0079$ \\
\cline{2-5}
decoherence & $H_{\Delta}\equiv0$ & $T_a''=-\frac{\ln (1-p_{0})}{\gamma_{0}+\gamma}$ & $T_{a}''=0.0101$ & $T_{a}''=0.0101$  \\
\hline
%\multirow{3}{*}{phase damping decoherence}
phase & $H_{\Delta}(x,y,z)$ & $T_{p}=\overline{T}$ & $T_{p}=0.0088$ & $T_{p}=0.0088$ \\
\cline{2-5}
damping & $\epsilon^{2}=2(\gamma_{0}+\gamma)^{2}$ & $T_{p}'=\frac{1-\sqrt{\bar{C}}}{2\sqrt{2}(\gamma_{0}+\gamma)}$ & $-$ & $T_{p}'=0.0090$  \\
\cline{2-5}
decoherence & $H_{\Delta}\equiv0$ & $T_{p}''=-\frac{\ln \bar{C}}{4(\gamma_{0}+\gamma)}$ & $T_{p}''=0.0128$ & $T_{p}''=0.0128$  \\
\hline
%\multirow{2}{*}{depolarizing decoherence}
depolarizing & $H_{\Delta}(x,y,z)$ & $T_{d}=-\frac{\ln (2\bar{P}-1)}{8(\gamma_{0}+\gamma)}$ & $T_{d}=0.0131$ & $T_{d}=0.0131$  \\
\cline{2-5}
decoherence & $H_{\Delta}\equiv0$ & $T_{d}=-\frac{\ln (2\bar{P}-1)}{8(\gamma_{0}+\gamma)}$ & $T_{d}=0.0131$ & $T_{d}=0.0131$  \\
\hline
\end{tabular}

\end{table}

\subsection{Illustrative examples}
\begin{example}[Sampling periods]
The values of sampling periods are shown in the right two columns of Table 1 for several specific cases, where we have assumed $p_{0}=0.01$, $\gamma_{0}=0.9$, $\gamma=0.1$ and $\bar{C}=\bar{P}=0.95$. Further, we can consider a real quantum system of a superconducting box in \cite{Wiseman and Milburn 2009}, \cite{Schuster et al 2005}. Let the resonance frequency $\tilde{\omega}_{0}=2\pi\times 100\text{MHz}$ and the cavity decay rate $\tilde{\gamma}_{0}=2\pi\times 0.8\text{MHz}$. Assume that $\tilde{\gamma}=\frac{2\pi\times 0.8\text{MHz}}{9}$ and $\tilde{\epsilon}=2\pi\times 1.0\text{MHz}$. Hence, the cavity decay time $T_{\gamma}=198.9\text{ns}$. Using the results in Table 1, we can get the real sampling periods as $\tilde{T}_{c}=8.0\text{ns}$, $\tilde{T}_{a}=1.0\text{ns}$, $\tilde{T}_{a}''=1.8\text{ns}$, $\tilde{T}_{p}=1.6\text{ns}$, $\tilde{T}_{p}''=2.3\text{ns}$ and $\tilde{T}_{d}=2.5\text{ns}$.
\end{example}

\begin{example}[Unitary control for Case A)]
Theorem \ref{CD_Theorem2} gives a sufficient condition for designing a unitary control to guarantee the required robustness. Here we employ a Lyapunov method \cite{Mirrahimi et al 2005}-\cite{Yi et al 2011} to design such a unitary control where the Lyapunov function is constructed based
on the Hilbert-Schmidt distance between a state $|\psi\rangle$ and
the sliding mode state $|0\rangle$; i.e.,
$V(|\psi\rangle, |0\rangle)=\frac{1}{2}(1-|\langle
0|\psi\rangle|^{2}).$ The control values can be selected as
(for details, see \cite{Dong and Petersen 2011Automatica},
\cite{Kuang and Cong 2008}):
\begin{equation}\label{Lyapunov_control}
 u_{k}=K_{k}f_{k}(\Im[e^{i\angle
\langle \psi|\phi_{j}\rangle} \langle 0|I_{k}|\psi\rangle]),
\ \ \ (k=x,y,z)
\end{equation}
where $\Im[a+bi]=b$ ($a, b\in \mathbf{R}$). Here $\angle c$ denotes the
argument of a complex number $c$, the parameter $K_{k}>0$ may be used to adjust
the control amplitude, and $f(\cdot)$ satisfies $xf(x)\geq 0$. We define
$\angle \langle \psi|0\rangle=0^{\circ}$ when $\langle
\psi|0\rangle=0$ and adopt the parameter values of $p_{0}=0.01$, $\epsilon=0.2$ and
$\beta=0.05$. From the simulation in \cite{Dong and Petersen
2011Automatica}, we find that the Lyapunov control is
not sensitive to small uncertainties in the system Hamiltonian. Additional
simulation results suggest that the robustness of the Lyapunov
control can be enhanced if we choose the terminal condition
$|\langle 1|\psi(t)\rangle|^{2}\leq \eta \alpha p_{0}$ (where
$0<\eta<1$) instead of $|\langle 1|\psi(t)\rangle|^{2}\leq \alpha
p_{0}$. Here, we select $\eta=0.8$. Hence, we design the
sampling period $T_{c}=1.0017$ using (\ref{period}). Using
Theorem \ref{CD_Theorem2}, we select $\alpha=2.5\times 10^{-3}$. We design the Lyapunov control using (\ref{Lyapunov_control}) and the terminal condition $|\langle 1|\psi(t)\rangle|^{2}\leq \eta \alpha
p_{0}=2.0\times 10^{-5}$ with
the control Hamiltonian $H_{u}=\frac{1}{2}u(t)\sigma_{y}$. Using
(\ref{Lyapunov_control}), we select $u(t)=K(\Im[e^{i\angle \langle
\psi(t)|0\rangle} \langle 0|\sigma_{y}|\psi(t)\rangle])$,
$K=500$, and let the time stepsize be $\delta t=10^{-6}$. We obtain the probability curve for $|0\rangle$ shown in Fig.
\ref{Lyapunovcontrol}(a) and the control value shown in Fig. \ref{Lyapunovcontrol}(b).
For the noise $\epsilon(t)I_x$ or $\epsilon(t)I_y$ where
$\epsilon(t)$ is a uniform distribution in $[-0.2, 0.2]$, additional
simulation results show that the state is also driven into
$\mathcal{E}_{c}$ using the Lyapunov control in Fig. \ref{Lyapunovcontrol}(b).

\begin{figure}
\centering
\includegraphics[width=4.6in]{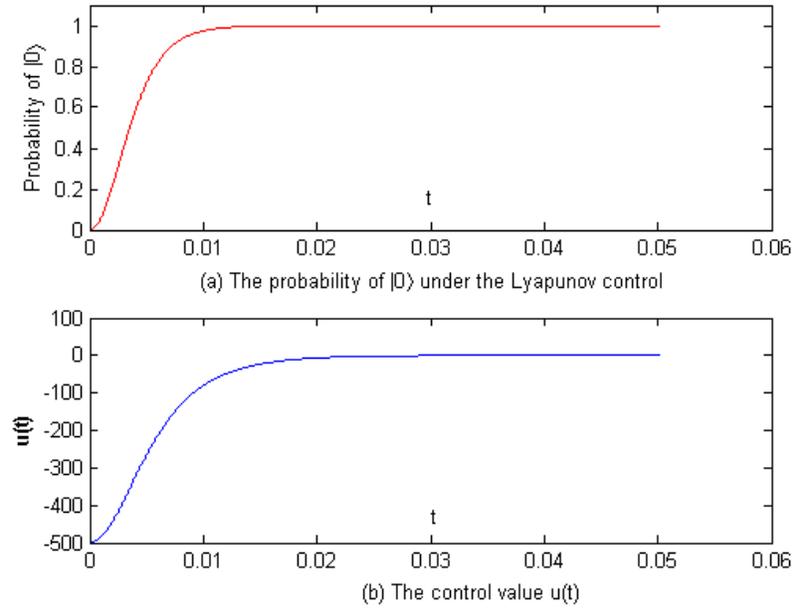}
\caption{The probability of $|0\rangle$ and the control value $u(t)$ under the Lyapunov control.} \label{Lyapunovcontrol}
\end{figure}

\end{example}

\begin{example}[Unitary control for Case B)]
For amplitude damping decoherence, Theorem
\ref{AD_Theorem2} gives a sufficient condition for designing a unitary control to guarantee the required robustness.  Here, we employ a constant control $H_{u}=\frac{1}{2}u \sigma_{y}$ ($u=6466$), and assume $p_{0}=0.01$, $\epsilon=0.2$, $\gamma_{0}=0.9$ and $\beta=0.05$. Using Theorem \ref{AD_Theorem2}, we may select $\alpha=0.05$. Let the time stepsize be $\delta t=10^{-7}$. The curve for $z_{t}=\text{tr}(\rho_{t}\sigma_{z})$ is shown in Fig. \ref{Example_case2}.

\begin{figure}
\centering
\includegraphics[width=3.6in]{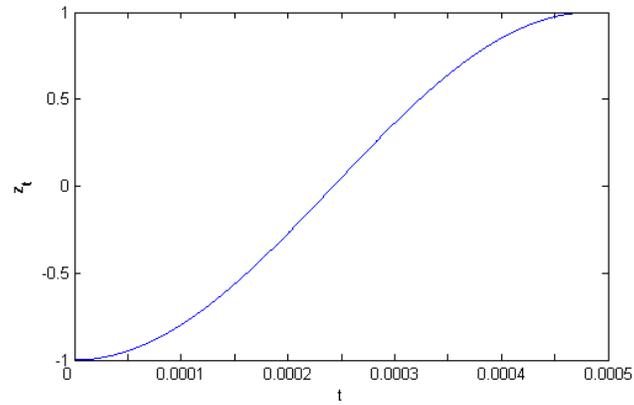}
\caption{The curve of $z_{t}$ for the case with amplitude damping decoherence, where $H_{u}=3233\sigma_{y}$.} \label{Example_case2}
\end{figure}

\end{example}

\begin{remark}
The required unitary control can be designed using strategies such as the Lyapunov method and optimal control theory. In Example 2 and Example 3, we used simulation to find appropriate control amplitudes for achieving the required objectives. Additional simulation experiments also show that $u(t)$ can tolerate small uncertainties. Since it is necessary to drive the state back to a subset of the sliding mode domain within a small time period, the required control amplitudes generally are relatively large, which is similar to the case of decoherence control based on dynamical decoupling \cite{Viola and Knill 2003}, \cite{Viola et al 1999}. The selection of small $\beta$ makes it reasonable that we first design the unitary control by ignoring possible uncertainties and then verify the robustness of the unitary control to uncertainties by simulation. Here, we present only simulated examples to demonstrate how such a unitary control can be designed. A systematic investigation into the design of the unitary control and finding optimal control amplitudes that can tolerate uncertainties will be the subject of future work.
\end{remark}

\section{Proof of the Main Results}\label{Sec4}
%\subsection{Proof of Theorem \ref{CD_Theorem1}}
%\begin{proof}
%Let $p_{0}'=\frac{p_{0}-p_{01}}{1-p_{01}-p_{10}}$. From Lemma \ref{lemma1}, we know if one makes a perfect projective %measurement at $t\leq T_{c}$ with $\sigma_{z}$ on the system, the probability of failure
%\begin{equation}\label{probability_comparison}
%p'=|\langle 1|\psi(t)\rangle|^{2}\leq p_{0}'.
%\end{equation}
%Assume $$|\psi(t)\rangle=a |0\rangle+b|1\rangle$$
%where $a$, $b$ are complex coefficients satisfying $|a|^{2}+|b|^{2}=1$. From (\ref{probability_comparison}), we have %$$|b|^{2}\leq p_{0}'$$
%Now we use the imperfect measurement model. The probability of failure
%\begin{equation}
%p=|\langle 1|\psi(t)\rangle|^{2}= |a|^{2}p_{01}+|b|^{2}(1-p_{10})\leq p_{01}+p_{0}'(1-p_{01}-p_{10})=p_{0}
%\end{equation}
%\end{proof}

\subsection{Proof of Theorem \ref{CD_Theorem2}}
To prove Theorem \ref{CD_Theorem2}, we first prove two lemmas (Lemma
\ref{lemmaB} and Lemma \ref{lemma2}). Lemma
\ref{lemmaB} compares the probabilities of failure for $H(t)=[1+\omega(t)]I_{z}+\epsilon\cos\phi_{0}I_{y}+\epsilon\sin\phi_{0}I_{x}$
and
$H(t)=\epsilon\cos\phi_{0}I_{y}+\epsilon\sin\phi_{0}I_{x}$. Lemma
\ref{lemmaB} together with Lemma \ref{lemma2} demonstrates that $H=\epsilon I_{x}$ can be used to estimate an upper bound on the probability of failure for $H(t)=[1+\omega(t)]I_{z}+\epsilon_{x}(t) I_{x}+\epsilon_{y}(t)
I_{y}$ when $z_{0}=1$.

\begin{lemma}\label{lemmaB}
For a single qubit with initial state
$(x_{0},y_{0},z_{0})=(0, 0, 1)$, the system evolves to
$(x^{A}_{t},y^{A}_{t},z^{A}_{t})$ and
$(x^{B}_{t},y^{B}_{t},z^{B}_{t})$ under the action of
$H^{A}=[1+\omega(t)]I_{z}+\epsilon\cos\phi_{0}I_{y}+\epsilon\sin\phi_{0}I_{x}$
(with constant $\epsilon> 0$ and $|\omega(t)|\leq \omega$) and
$H^{B}=\epsilon\cos\phi_{0}I_{y}+\epsilon\sin\phi_{0}I_{x}$,
respectively. For arbitrary $t\in [0,
\frac{\pi}{2\sqrt{4+\epsilon^{2}}}]$, $z^{A}_{t}\geq z^{B}_{t}$.
\end{lemma}

\begin{proof}
For the system with Hamiltonian $H^{A}=[1+\omega
(t)]I_{z}+\epsilon\cos\phi_{0}I_{y}+\epsilon\sin\phi_{0}I_{x}$,
using (\ref{blochEq}) and (\ref{Case1Eq}), we obtain the
following state equations
\begin{equation}\label{lemmaB1}
\left(%
\begin{array}{c}
  \dot{x}_{t}^{A} \\
  \dot{y}_{t}^{A} \\
  \dot{z}_{t}^{A} \\
\end{array}%
\right)
=\left(%
\begin{array}{ccc}
  0 & -[1+\omega(t)] & \epsilon\cos\phi_{0} \\
  1+\omega(t) & 0  & -\epsilon\sin\phi_{0} \\
  -\epsilon\cos\phi_{0} & \epsilon\sin\phi_{0} & 0 \\
\end{array}%
\right) \left(%
\begin{array}{c}
  x_{t}^{A} \\
  y_{t}^{A} \\
  z_{t}^{A} \\
\end{array}%
\right),
\end{equation}
where $(x_{0}^{A}, y_{0}^{A}, z_{0}^{A})=(0, 0, 1)$.

Consider $\omega(t)$ as a control input and select the
performance measure as
\begin{equation}
J(\omega)=z_{f} .
\end{equation}
We introduce the Lagrange multiplier vector
$\lambda(t)=(\lambda_{1}(t), \lambda_{2}(t), \lambda_{3}(t))^{T}$
and obtain the corresponding Hamiltonian function as follows:
\begin{equation}
\mathbb{H}({\mathbf{r}(t),\omega(t),\mathbf{\lambda}(t),t})\equiv
\lambda^{T}(t)\left(%
\begin{array}{ccc}
  0 & -[1+\omega(t)] & \epsilon\cos\phi_{0} \\
  1+\omega(t) & 0  & -\epsilon\sin\phi_{0} \\
  -\epsilon\cos\phi_{0} & \epsilon\sin\phi_{0} & 0 \\
\end{array}%
\right) \left(%
\begin{array}{c}
  x_{t} \\
  y_{t} \\
  z_{t} \\
\end{array}%
\right),
\end{equation}
where $\mathbf{r}(t)=(x_{t}, y_{t}, z_{t})$. That is
\begin{equation}
\begin{array}{c}
\mathbb{H}({\mathbf{r}(t),\omega(t),\mathbf{\lambda}(t),t}) \ \ \ \ \ \ \ \  \ \ \ \  \ \
\ \ \ \ \ \ \ \ \ \ \ \ \ \ \ \ \ \ \ \ \ \ \ \ \ \ \ \ \ \ \ \ \ \ \ \ \ \ \ \ \ \ \ \ \ \
\ \ \ \ \ \ \ \ \ \ \ \ \ \ \ \ \ \ \ \ \ \ \ \ \ \ \ \ \ \ \ \\
=[1+\omega(t)](\lambda_{2}(t)x_{t}-\lambda_{1}(t)y_{t})+\epsilon
\cos\phi_{0}(\lambda_{1}(t)z_{t}-\lambda_{3}(t)x_{t})-\epsilon
\sin\phi_{0}(\lambda_{2}(t)z_{t}-\lambda_{3}(t)y_{t}).
\end{array}
\end{equation}
According to Pontryagin's minimum principle \cite{Kirk 1970}, a
necessary condition for $\omega^{*}(t)$ to minimize
%the functional
$J(\omega)$ is
\begin{equation}
\mathbb{H}({\mathbf{r}^{*}(t),\omega^{*}(t),\mathbf{\lambda}^{*}(t),t})\leq
\mathbb{H}({\mathbf{r}^{*}(t),\omega(t),\mathbf{\lambda}^{*}(t),t})
.
\end{equation}
Hence, if we do not consider singular cases (i.e.,
$\lambda_{2}(t)x_{t}-\lambda_{1}(t)y_{t}\equiv 0$), the optimal
control $\omega^{*}(t)$ should be chosen as follows:
\begin{equation}\label{bangbang}
\omega^{*}(t)=-\omega
\text{sgn}(\lambda_{2}(t)x_{t}-\lambda_{1}(t)y_{t}).
\end{equation}
That is, the optimal control strategy for $\omega(t)$ is bang-bang
control; i.e., $\omega^{*}(t)=\bar{\omega}=+\omega \ \text{or}
-\omega$. Now we consider
$H^{A}=(1+\bar{\omega})I_{z}+\epsilon\cos\phi_{0}I_{y}+\epsilon\sin\phi_{0}I_{x}$,
which leads to the following state equations
\begin{equation}
\left(%
\begin{array}{c}
  \dot{x}_{t}^{A} \\
  \dot{y}_{t}^{A} \\
  \dot{z}_{t}^{A} \\
\end{array}%
\right)
=\left(%
\begin{array}{ccc}
  0 & -(1+\bar{\omega}) & \epsilon\cos\phi_{0} \\
  1+\bar{\omega} & 0  & -\epsilon\sin\phi_{0} \\
  -\epsilon\cos\phi_{0} & \epsilon\sin\phi_{0} & 0 \\
\end{array}%
\right) \left(%
\begin{array}{c}
  x_{t}^{A} \\
  y_{t}^{A} \\
  z_{t}^{A} \\
\end{array}%
\right),
\end{equation}
where $(x_{0}^{A}, y_{0}^{A}, z_{0}^{A})=(0, 0, 1)$. The
corresponding solution is
\begin{equation}\label{Asolution}
\left(%
\begin{array}{c}
  x_{t}^{A} \\
  y_{t}^{A} \\
  z_{t}^{A} \\
\end{array}%
\right)
=\\
\left(%
\begin{array}{c}
  \frac{\epsilon\cos\phi_{0}}{\sqrt{(1+\bar{\omega})^{2}+\epsilon^{2}}}\sin\upsilon t
 -\frac{(1+\bar{\omega})\epsilon\sin\phi_{0}}{(1+\bar{\omega})^{2}+\epsilon^{2}}\cos\upsilon t
  +\frac{(1+\bar{\omega})\epsilon\sin\phi_{0}}{(1+\bar{\omega})^{2}+\epsilon^{2}} \\
  -\frac{\epsilon\sin\phi_{0}}{\sqrt{(1+\bar{\omega})^{2}+\epsilon^{2}}}\sin\upsilon
  t-\frac{(1+\bar{\omega})\epsilon\cos\phi_{0}}{(1+\bar{\omega})^{2}+\epsilon^{2}}\cos\upsilon
  t+\frac{(1+\bar{\omega})\epsilon\cos\phi_{0}}{(1+\bar{\omega})^{2}+\epsilon^{2}} \\
  \frac{\epsilon^{2}}{(1+\bar{\omega})^{2}+\epsilon^{2}}\cos\upsilon t+\frac{(1+\bar{\omega})^{2}}{(1+\bar{\omega})^{2}+\epsilon^{2}} \\
\end{array}%
\right),
\end{equation}
where $\upsilon=\sqrt{(1+\bar{\omega})^{2}+\epsilon^{2}}$.
From (\ref{Asolution}), we know that $z_t$ is a monotonically
decreasing function in $t$ when $t\in [0, \frac{\pi}{2
\sqrt{4+\epsilon^{2}}}]$. Hence, we only consider the case $t \in
[0,t_{f}]$ where $t_{f}\in [0, \frac{\pi}{2
\sqrt{4+\epsilon^{2}}}]$.

Now consider the optimal control problem with a fixed final time
$t_{f}$ and a free final state $\mathbf{r}_{f}=(x_{f},y_{f},z_{f})$.
According to Pontryagin's minimum principle,
$\lambda^{*}(t_f)=\frac{\partial}{\partial
\mathbf{r}}\mathbf{r}^{*}(t_{f})$, and it is straightforward
to verify that
$(\lambda_{1}(t_f),\lambda_{2}(t_f),\lambda_{3}(t_f))=(0,0,1)$. Now
consider another necessary condition
$\dot{\lambda}(t)=-\frac{\partial
\mathbb{H}({\mathbf{r}(t),\epsilon(t),\mathbf{\lambda}(t),t})}{\partial
\mathbf{r}}$ which leads to the following relationships:

\begin{equation}\label{eq24b}
\dot{\mathbf{\lambda}}(t)=\left(%
\begin{array}{c}
  \dot{\lambda}_{1}(t) \\
  \dot{\lambda}_{2}(t) \\
  \dot{\lambda}_{3}(t) \\
\end{array}%
\right)
=\left(%
\begin{array}{ccc}
  0 & -(1+\bar{\omega}) & \epsilon\cos\phi_{0} \\
  1+\bar{\omega} & 0  & -\epsilon\sin\phi_{0} \\
  -\epsilon\cos\phi_{0} & \epsilon\sin\phi_{0} & 0 \\
\end{array}%
\right) \left(%
\begin{array}{c}
  \lambda_{1}(t) \\
  \lambda_{2}(t) \\
  \lambda_{3}(t) \\
\end{array}%
\right),
\end{equation}
where
$(\lambda_{1}(t_f),\lambda_{2}(t_f),\lambda_{3}(t_f))=(0,0,1)$. The
corresponding solution is
\begin{equation}
\left(%
\begin{array}{c}
  \lambda_{1}(t) \\
  \lambda_{2}(t) \\
  \lambda_{3}(t) \\
\end{array}%
\right)
=\left(%
\begin{array}{c}
  -\frac{\epsilon\cos\phi_{0}}{\sqrt{(1+\bar{\omega})^{2}+\epsilon^{2}}}\sin\upsilon(t_{f}-t)
 -\frac{(1+\bar{\omega})\epsilon\sin\phi_{0}}{(1+\bar{\omega})^{2}+\epsilon^{2}}\cos\upsilon(t_{f}-t)
  +\frac{(1+\bar{\omega})\epsilon\sin\phi_{0}}{(1+\bar{\omega})^{2}+\epsilon^{2}} \\
  \frac{\epsilon\sin\phi_{0}}{\sqrt{(1+\bar{\omega})^{2}+\epsilon^{2}}}\sin\upsilon(t_{f}-t)
  -\frac{(1+\bar{\omega})\epsilon\cos\phi_{0}}{(1+\bar{\omega})^{2}+\epsilon^{2}}\cos\upsilon(t_{f}-t)
  +\frac{(1+\bar{\omega})\epsilon\cos\phi_{0}}{(1+\bar{\omega})^{2}+\epsilon^{2}} \\
  \frac{\epsilon^{2}}{(1+\bar{\omega})^{2}+\epsilon^{2}}\cos\upsilon(t_{f}-t)+\frac{(1+\bar{\omega})^{2}}{(1+\bar{\omega})^{2}+\epsilon^{2}} \\
\end{array}%
\right).
\end{equation}
We obtain
\begin{equation}\label{singularcondition}
\lambda_{2}(t)x_{t}-\lambda_{1}(t)y_{t}
=\frac{\epsilon^{2}(1+\bar{\omega})}{\upsilon^{3/2}}[\sin\upsilon
t+\sin\upsilon(t_{f}-t)-\sin\upsilon t_f].
\end{equation}
It is easy to show that the quantity
$(\lambda_{2}(t)x_{t}-\lambda_{1}(t)y_{t})\geq 0$ occurring in
(\ref{bangbang}) does not change sign when $t_{f}\in [0,
\frac{\pi}{2 \sqrt{4+\epsilon^{2}}}]$ and $t \in [0,t_{f}]$. Hence,
the optimal control is $\delta^{*}(t)=\bar{\omega}=-\omega$.

We now exclude the possibility that there exists a singular
case. Suppose that there exists a singular interval $[t_{0}, t_{1}]$
(where $t_{0}\geq 0$ and we assume that $[t_{0}, t_{1}]$ is the
first singular interval) such that when $t\in [t_{0}, t_{1}]$
\begin{equation}\label{singular1Eq}
h(t)=\lambda_{2}(t)x_{t}-\lambda_{1}(t)y_{t}\equiv 0.
\end{equation}
We also have the following relationship
\begin{equation}\label{singular2Eq}
\ddot{h}(t)=\lambda_{3}(t)x_{t}-\lambda_{1}(t)z_{t}\equiv 0
\end{equation}
where we have used (\ref{lemmaB1}) and the following costate
equation
\begin{equation}
\dot{\mathbf{\lambda}}(t)=\left(%
\begin{array}{c}
  \dot{\lambda}_{1}(t) \\
  \dot{\lambda}_{2}(t) \\
  \dot{\lambda}_{3}(t) \\
\end{array}%
\right)
=\left(%
\begin{array}{ccc}
  0 & -[1+\omega(t)] & \epsilon\cos\phi_{0} \\
  1+\omega(t) & 0  & -\epsilon\sin\phi_{0} \\
  -\epsilon\cos\phi_{0} & \epsilon\sin\phi_{0} & 0 \\
\end{array}%
\right) \left(%
\begin{array}{c}
  \lambda_{1}(t) \\
  \lambda_{2}(t) \\
  \lambda_{3}(t) \\
\end{array}%
\right).
\end{equation}

If $t_{0}=0$, we have $(x_{0}, y_{0}, z_{0})=(0, 0, 1)$. By the
principle of optimality \cite{Kirk 1970}, we may consider the case
$t_f=t_1$. Using (\ref{singular1Eq}), (\ref{singular2Eq}) and
$(\lambda_{1}(t_1),\lambda_{2}(t_1),\lambda_{3}(t_1))=(0,0,1)$, we
have $x_{t_1}=0$ and $y_{t_1}=0$. Using the relationship of
$x^{2}_{t}+y^{2}_{t}+z_{t}^{2}=1$, we obtain $z_{t_{1}}=1$ or $-1$.
If $z_{t_{1}}=1$, the initial and final states are the same
state $|0\rangle$. However, if we use the control
$\omega(t)=\bar{\omega}$, from (\ref{Asolution}) we have
$z_{t_1}(\bar{\omega})=\frac{\epsilon^2}{(1+\bar\omega)^{2}+\epsilon^2}\cos\upsilon
t_1+\frac{(1+\bar\omega)^{2}}{(1+\bar\omega)^{2}+\epsilon^2}<z_{t_1}=1$.
Hence, this contradicts the fact that we are considering the optimal
case $\min z_f$. If $z_{t_{1}}=-1$, there exists $0< \tilde{t}_1<
t_1$ such that $z_{\tilde{t}_1}=0$. By the principle of optimality
\cite{Kirk 1970}, we may consider the case $t_f=\tilde{t}_1$. From
the two equations (\ref{singular1Eq}) and (\ref{singular2Eq}), we
know that $z^{2}_{\tilde{t}_1}=1$ which contradicts
$z_{\tilde{t}_1}=0$. Hence, no singular condition can exist if
$t_{0}=0$.

If $t_{0}> 0$, using (\ref{bangbang}) we must select
$\omega(t)=\bar{\omega}$ when $t\in [0, t_{0}]$. From
(\ref{singularcondition}), we know that there exist no $t_{0}\in (0,
t_{f})$ satisfying
$\lambda_{2}(t_{0})x_{t_{0}}-\lambda_{1}(t_{0})y_{t_{0}}=0$. Hence,
there exist no singular cases for our problem. From the previous
analysis, $\omega(t)=-\omega$ is the optimal control when $t\in [0,
\frac{\pi}{2 \sqrt{4+\epsilon^{2}}}]$.

For the system with Hamiltonian
$H^{B}=\epsilon\cos\phi_{0}I_{y}+\epsilon\sin\phi_{0}I_{x}$,
using (\ref{blochEq}) and (\ref{Case1Eq}), we obtain the
following state equations
\begin{equation}\label{HB}
\left(%
\begin{array}{c}
  \dot{x}_{t}^{B} \\
  \dot{y}_{t}^{B} \\
  \dot{z}_{t}^{B} \\
\end{array}%
\right)
=\left(%
\begin{array}{ccc}
  0 & 0 & \epsilon\cos\phi_{0} \\
  0 & 0  & -\epsilon\sin\phi_{0} \\
  -\epsilon\cos\phi_{0} & \epsilon\sin\phi_{0} & 0 \\
\end{array}%
\right) \left(%
\begin{array}{c}
  x_{t}^{B} \\
  y_{t}^{B} \\
  z_{t}^{B} \\
\end{array}%
\right),
\end{equation}
where $(x_{0}^{B}, y_{0}^{B}, z_{0}^{B})=(0, 0, 1)$. The
corresponding solution is
\begin{equation}\label{Hbsolution}
\left(%
\begin{array}{c}
  x_{t}^{B} \\
  y_{t}^{B} \\
  z_{t}^{B} \\
\end{array}%
\right)
=\\
\left(%
\begin{array}{c}
  \cos\phi_{0}\sin\epsilon t \\
  -\sin\phi_{0}\sin\epsilon t \\
  \cos\epsilon t \\
\end{array}%
\right).
\end{equation}

We define $F(t)$ and $f(t)$ as follows:
\begin{equation}
F(t)=z_{t}^{A}-z_{t}^{B}=\frac{\epsilon^{2}}{(1-\omega)^{2}+\epsilon^{2}}\cos\upsilon
t+\frac{(1-\omega)^{2}}{(1-\omega)^{2}+\epsilon^{2}}-\cos \epsilon
t,
\end{equation}
\begin{equation}\label{eq3}
f(t)=\dot{F}(t)=-\frac{\epsilon^{2}}{\sqrt{(1-\omega)^{2}+\epsilon^{2}}}\sin\upsilon
t+\epsilon\sin \epsilon t .
\end{equation}
Now, consider $t\in [0, \frac{\pi}{2\sqrt{4+\epsilon^{2}}}]$ to
obtain
\begin{equation}\label{eq5}
\dot{f}(t)=\epsilon^{2}(\cos\epsilon t-\cos\upsilon t)\geq 0.
\end{equation}
It is clear that $\dot{f}(t)=0$ only when $t=0$. Hence $f(t)$ is a
monotonically increasing function and $$\min_t{f(t)}=f(0)=0 .$$

Hence, we have
\begin{equation}\label{eq6}
f(t)\geq 0 .
\end{equation}
From this result, it is clear that $F(t)$ is a monotonically increasing
function and
$$\min_t{F(t)}=F(0)=0 .$$
Hence $F(t)\geq 0$ when $t\in [0,
\frac{\pi}{2\sqrt{4+\epsilon^{2}}}]$. Therefore,  we can conclude that
$z^{A}_{t}\geq z^{B}_{t}$ for arbitrary $t\in [0,
\frac{\pi}{2\sqrt{4+\epsilon^{2}}}]$.
\end{proof}

We now present another lemma.

\begin{lemma}\label{lemma2}
For a single qubit with initial state
$(x_{0},y_{0},z_{0})=(0, 0, 1)$, suppose that the system evolves to
$(x_{t},y_{t},z_{t})$ under the action of $H=\epsilon(\cos\phi
I_{y}+\sin\phi I_{x})$ ($\phi$ is a constant). Then, $z_{t}$ is
independent of $\phi$.
\end{lemma}

\begin{proof}
For $H=\omega(\sin\phi I_{x}+\cos\phi I_{y})$, from
(\ref{Hbsolution}), we have
$$z_{t}=\cos\epsilon t .$$
It is clear that $z_{t}$ is independent of  $\phi$.
\end{proof}

\begin{remark}
Since $z_{t}$ is independent of $\phi$, it is enough to consider
a special case $\phi=\frac{\pi}{2}$ when analyzing $z_{t}$ under
$H=\epsilon (\cos\phi I_{y}+\sin\phi I_{x})$.
\end{remark}

Now we can prove Theorem \ref{CD_Theorem2}.

\begin{proof}
For a single qubit, assume that the state at time $t$ is $\rho_{t}$. If we make a measurement with
the operator $\sigma_{z}$, the probability $p$ that the state
will collapse into $|1\rangle$ (the probability of failure) is
\begin{equation}\label{fidelity}
p=\langle 1|\rho_{t}|1\rangle=\frac{1-z_{t}}{2} .
\end{equation}
For a closed single qubit system, its state $|\psi\rangle$ can be represented as
\begin{equation}\label{superposition}
|\psi\rangle=\cos{\frac{\theta}{2}}|0\rangle+e^{i\varphi}\sin{\frac{\theta}{2}}|1\rangle
,
\end{equation}
where its bloch vector
corresponds to $(x, y, z)=(\sin\theta \cos\varphi,
\sin\theta \sin\varphi, \cos\theta)$, $\theta \in [0,\pi]$, $\varphi
\in [0, 2\pi]$.

For $H^{A}=[1+\omega(t)]I_{z}+\epsilon_{x}(t) I_{x}+\epsilon_{y}(t)
I_{y}$, using $\dot{\rho}=-i[H^{A}, \rho]$ and (\ref{blochEq}), we
obtain the following state equations
\begin{equation}
\left(%
\begin{array}{c}
  \dot{x}_{t}^{A} \\
  \dot{y}_{t}^{A} \\
  \dot{z}_{t}^{A} \\
\end{array}%
\right)
=\left(%
\begin{array}{ccc}
  0 & -[1+\omega(t)] & \epsilon_{y}(t) \\
  1+\omega(t) & 0  & -\epsilon_{x}(t) \\
  -\epsilon_{y}(t) & \epsilon_{x}(t) & 0 \\
\end{array}%
\right) \left(%
\begin{array}{c}
  x_{t}^{A} \\
  y_{t}^{A} \\
  z_{t}^{A} \\
\end{array}%
\right).
\end{equation}
%We first consider $z_{0}=\cos\theta_{0}= 1-2\alpha p_{0}$, where
%$\theta_{0}\in (0, \frac{\pi}{2})$.
Define $\epsilon(t)=\sqrt{\epsilon_{x}^{2}(t)+\epsilon_{y}^{2}(t)}$
and $\epsilon_{x}(t)=\epsilon(t)\sin\phi_{t}$ ,
$\epsilon_{y}(t)=\epsilon(t)\cos\phi_{t}$. This leads to the
following equation
\begin{equation}\label{twouncertainties}
\left(%
\begin{array}{c}
  \dot{x}_{t}^{A} \\
  \dot{y}_{t}^{A} \\
  \dot{z}_{t}^{A} \\
\end{array}%
\right)=
\left(%
\begin{array}{ccc}
  0 & -[1+\omega(t)] & \epsilon(t) \cos\phi_{t} \\
  1+\omega(t) & 0  & -\epsilon(t) \sin\phi_{t} \\
  -\epsilon(t) \cos\phi_{t} & \epsilon(t) \sin\phi_{t} & 0 \\
\end{array}%
\right) \left(%
\begin{array}{c}
  x_{t}^{A} \\
  y_{t}^{A} \\
  z_{t}^{A} \\
\end{array}%
\right).
\end{equation}
%where $(x_{0}^{A}, y_{0}^{A},
%z_{0}^{A})=(\sin\theta_{0}\cos\varphi_{0},
%\sin\theta_{0}\sin\varphi_{0}, \cos\theta_{0})$ and $\varphi_{0} \in
%[0, 2\pi]$.
%Define $N(t)=-\epsilon
%(t)\sin\theta_{0}\cos(\phi_{t}+\varphi_{0})$. From
%(\ref{twouncertainties}), we have
%\begin{equation}\label{differential}
%\dot{z}_t^{A}|_{t=0}=\lim_{t\rightarrow 0}N(t).
%\end{equation}
%Now we take an arbitrary evolution state (except $|1\rangle$)
%starting from $|0\rangle$ as a new initial state.
For
$H^{B}=\epsilon I_x$,
%the initial state can be represented as
%$(x'_{0}, y'_{0}, z_{0})=(0, -\sin\theta_{0}, \cos\theta_{0})$,
%where $\theta_{0} \in (0, \pi)$.
we have
\begin{equation}\label{eq20}
\left(%
\begin{array}{c}
  \dot{x}_{t}^{B} \\
  \dot{y}_{t}^{B} \\
  \dot{z}_{t}^{B} \\
\end{array}%
\right)
=\left(%
\begin{array}{ccc}
  0 & 0 & 0 \\
  0 & 0  & -\epsilon \\
  0 & \epsilon & 0 \\
\end{array}%
\right) \left(%
\begin{array}{c}
  x_{t}^{B} \\
  y_{t}^{B} \\
  z_{t}^{B} \\
\end{array}%
\right).
\end{equation}
When $(x_{0}^{A}, y_{0}^{A},
z_{0}^{A})=(x_{0}^{B}, y_{0}^{B},
z_{0}^{B})=(0, 0, 1)$, for $\Delta t\rightarrow 0$, we have from Lemma \ref{lemmaB} and Lemma
\ref{lemma2}
\begin{equation}
z_{\Delta t}^{A}\geq z_{\Delta t}^{B}.
\end{equation}
We will now prove that the relationship $z_{\Delta t}^{A}\geq z_{\Delta t}^{B}$ ($\Delta t\rightarrow 0$) is also true for $z_{0}^{A}=z_{0}^{B}=\cos\theta_{0}$ (where $\theta_{0}\in (0, \pi)$).
%When $z_{0}^{A}=z_{0}^{B}=\cos\theta_{0}$, we can assume $(x_{0}^{A}, y_{0}^{A},
%z_{0}^{A})=(\sin\theta_{0}\cos\varphi_{0},
%\sin\theta_{0}\sin\varphi_{0}, \cos\theta_{0})$ and $(x_{0}^{B}, y_{0}^{B}, z_{0}^{B})=(0, -\sin\theta_{0}, %\cos\theta_{0})$ (where $\varphi_{0} \in
%[0, 2\pi]$).
We assume that there exist $\tilde{t}\in(0, \Delta t]$ such that
\begin{equation}
z_{\tilde{t}}^{A}< z_{\tilde{t}}^{B}.
\end{equation}
Define $f(t)=z_{t}^{A}-z_{t}^{B}$. Since $f(t)$ is continuous in
$t$ and $f(0)=0$, there exists a time $t^{*}=\sup \{t|0 \leq
t<\tilde{t}, f(t)=0\}$ satisfying $f(t^{*})=0$ and $f(t)<0$ for $t\in (t^{*}, \tilde{t}]$. Hence
\begin{equation}\label{FtDot}
\dot{f}(t)|_{t=t^{*}}\leq 0.
\end{equation}

Let $z_{t^{*}}^{A}=z_{t^{*}}^{B}=\cos\theta^{*}$. We can assume $(x_{t^{*}}^{A}, y_{t^{*}}^{A},
z_{t^{*}}^{A})=(\sin\theta^{*}\cos\varphi^{*},
\sin\theta^{*}\sin\varphi^{*}, \cos\theta^{*})$ and $(x_{t^{*}}^{B}, y_{t^{*}}^{B}, z_{t^{*}}^{B})=(0, -\sin\theta^{*}, \cos\theta^{*})$ (where $\varphi^{*} \in
[0, 2\pi]$).
Define $N(t)=-\epsilon(t)\sin\theta^{*}\cos(\phi_{t}+\varphi^{*})$. From
(\ref{twouncertainties}) and (\ref{eq20}), we have
\begin{equation}\label{differential}
\dot{f}(t)|_{t=t^{*}}=\dot{z}_t^{A}|_{t=t^{*}}-\dot{z}^{B}_t|_{t=t^{*}}=\lim_{t\rightarrow t^{*}}N(t)+\epsilon\sin\theta^{*}.
\end{equation}
For arbitrary $t$, it is clear that
\begin{equation} N(t)\geq-\epsilon\sin\theta^{*}. \end{equation}
When $N(t)> -\epsilon\sin\theta^{*}$, $\dot{f}(t)|_{t=t^{*}}> 0$, which contradicts (\ref{FtDot}). When $ N(t)=-\epsilon\sin\theta^{*}$, since $\sin\theta^{*}\neq 0$, we have $\epsilon(t^{*})=\epsilon$ and $\phi_{t^{*}}=2\pi-\varphi^{*}$. Using Pontryagin's minimum principle \cite{Kirk 1970} and a similar argument in Lemma \ref{lemmaB} and Lemma \ref{lemma2}, we can prove $z^{A}_{t^{*}+\Delta t}\geq z^{B}_{t^{*}+\Delta t}$. Hence we can conclude that for $z_{0}^{A}=z_{0}^{B}=\cos\theta_{0}$ (where $\theta_{0}\in[0, \pi)$) and $\Delta t\rightarrow 0$,
\begin{equation}
z^{A}_{\Delta t}\geq z^{B}_{\Delta t}.
\end{equation}
%When $\sin\theta_{0}=0$, using Lemma \ref{lemmaB} and Lemma
%\ref{lemma2}, we can also obtain the same conclusion $ z^{A}_{\Delta
%t}\geq z^{B}_{\Delta t}$.

From (\ref{eq20}), we know that $z^{B}_{t}=\cos(\theta_{0}+\epsilon
t)$. When $0< t< \frac{\pi-\theta_{0}}{\epsilon}$, $z^{B}_{t}$
decreases monotonically in $t$. We now define $g(t)=z^{A}_{t}-
z^{B}_{t}$ and assume that there exist $t=t_{1}\in [0,
\frac{\pi-\theta_{0}}{\epsilon})$ such that $z^{A}_{t_{1}}<
z^{B}_{t_{1}}$. That is, $g(t_{1})<0$. Since $g(t)$ is continuous in
$t$ and $g(0)=0$, there exists a time $t^{*}=\sup \{t|0 \leq
t<t_{1}, g(t)=0\}$ satisfying $g(t)<0$ for $t\in (t^{*}, t_{1}]$.
However, we have established that for any $z_{t}^{A}=z_{t}^{B}$ and
$\Delta t \rightarrow 0$, $z_{t+\Delta t}^{A}\geq z_{t+\Delta
t}^{B}$, which contradicts $g(t)<0$ for $t\in (t^{*}, t_{1}]$.
Hence, we have the following relationship for $t\in [0,
\frac{\pi-\theta_{0}}{\epsilon})$
\begin{equation}\label{ZaZbRelationship}
z^{A}_{t}\geq z^{B}_{t}.
\end{equation}
From (\ref{fidelity}), it is clear that the probabilities of failure
satisfy $p^{A}_{t}=\frac{1-z^{A}_{t}}{2}\leq
p^{B}_{t}=\frac{1-z^{B}_{t}}{2}$. That is, the probability of
failure $p_{t}^{A}$ is not greater than $p^{B}_{t}$ for $t\in [0,
\frac{\pi-\theta_{0}}{\epsilon})$.

Since $z_{t}^{B}=\cos(\theta_{0}+\epsilon t)$, we have $\Delta
z^{B}_{\beta T_{c}}=\cos\theta_{0}-\cos(\theta_{0}+\epsilon \beta T)$,
where
\begin{equation}
T_{c}=\frac{\arccos(1-2p_{0})}{\epsilon}.
\end{equation}
When $|\langle \psi(0)|1\rangle|^{2}\leq \alpha p_{0}$, using the previous argument, we have
$$z^{A}_{\beta T_{c}}\geq 1-2\alpha p_{0}+\cos(\theta_{0}+\epsilon \beta T_{c})-\cos\theta_{0}=M.$$
Now let $$p=\frac{1-z^{A}_{\beta T_{c}}}{2}\leq \frac{1-M}{2}\leq p_{0}
.$$ Using the fact $\theta_{0}=\arccos(1-2\alpha p_{0})$, we have
the following relationship
\begin{equation}
\alpha \leq\frac{1-\cos[(1-\beta)\arccos(1-2p_{0})]}{2p_{0}} .
\end{equation}
\end{proof}

%If we select a performance index
%$J=\frac{1}{T_{0}}+p_{fail}\int^{T_{1}}_{0}u^{2}(t)dt$, we can
%optimize the cost.
\subsection{Proof of Theorem \ref{AC_Thereom1}}
\begin{proof}
For the open qubit system subject to (\ref{AD_masterEq}),
%with Hamiltonian $H(t)$ and amplitude damping decoherence with coupling strength $\gamma_{t}$, its dynamics can be %described as follows
%\begin{equation}
%\dot{\rho_t}=-i[H(t), \rho_t]+\gamma_{t}(\sigma_{-}\rho_t %\sigma_{+}-\frac{1}{2}\sigma_{+}\sigma_{-}\rho_t-\frac{1}{2}\rho_{t}\sigma_{+}\sigma_{-})
%\end{equation}
when $H(t)=[1+\omega(t)]I_{z}+\epsilon_{x}(t)I_{x}+\epsilon_{y}(t)I_{y}$ ($\sqrt{\epsilon_{x}^{2}(t)+\epsilon_{y}^{2}(t)}\leq
\epsilon$, $\epsilon >0$, $|\omega(t)|\leq \omega$ and $\omega\geq 0$), $\gamma_{t}=\gamma_{0}+\delta \gamma_{t}$ ($|\delta \gamma_{t}|\leq \gamma$), using (\ref{blochEq}), we have
\begin{equation}\label{AD_stateEq}
\left(%
\begin{array}{c}
  \dot{x}_{t} \\
  \dot{y}_{t} \\
  \dot{z}_{t} \\
\end{array}%
\right)
=\left(%
\begin{array}{ccc}
  -\frac{1}{2}(\gamma_{0}+\delta \gamma_{t}) & -(1+\omega(t)) & \epsilon_{y}(t) \\
  1+\omega(t) & -\frac{1}{2}(\gamma_{0}+\delta \gamma_{t})  & -\epsilon_{x}(t) \\
  -\epsilon_{y}(t) & \epsilon_{x}(t) & -(\gamma_{0}+\delta \gamma_{t}) \\
\end{array}%
\right) \left(%
\begin{array}{c}
  x_{t} \\
  y_{t} \\
  z_{t} \\
\end{array}%
\right)+\left(%
\begin{array}{c}
  0 \\
  0 \\
  -(\gamma_{0}+\delta \gamma_{t}) \\
\end{array}%
\right),
\end{equation}
where $(x_{0}, y_{0}, z_{0})=(0, 0, 1)$. From (\ref{AD_stateEq}), we have
\begin{equation}
\begin{array}{cc}
\dot{z_{t}}& =-\epsilon_{y}(t)x_{t}+\epsilon_{x}(t)y_{t}-(\gamma_{0}+\delta \gamma_{t})(z_{t}+1)\\
& \geq -2\epsilon \sqrt{1-z_{t}^{2}}-(\gamma_{0}+\gamma)(z_{t}+1).\ \ \ \ \ \ \
\end{array}
\end{equation}
Denoting
\begin{equation}
f(z)=2\epsilon \sqrt{1-z_{t}^{2}}+(\gamma_{0}+\gamma)(1+z_{t}),
\end{equation}
we have
\begin{equation}
\frac{\text{d}f(z)}{\text{d}z}=(\gamma_{0}+\gamma)-2\epsilon\frac{z_{t}}{\sqrt{1-z_{t}^{2}}}.
\end{equation}
Let
$\frac{\text{d}f(z)}{\text{d}z}=0$ to find the solution $z=\frac{\gamma_{0}+\gamma}{\sqrt{4\epsilon^{2}+(\gamma_{0}+\gamma)^{2}}}$.
Hence,
\begin{equation}
\max f(z)=f(\frac{\gamma_{0}+\gamma}{\sqrt{4\epsilon^{2}+(\gamma_{0}+\gamma)^{2}}})
=\sqrt{4\epsilon^{2}+(\gamma_{0}+\gamma)^{2}}+(\gamma_{0}+\gamma).
\end{equation}
Hence,
\begin{equation}
\dot{z_{t}}\geq -\max f(z)=-\sqrt{4\epsilon^{2}+(\gamma_{0}+\gamma)^{2}}-(\gamma_{0}+\gamma).
\end{equation}
When $t\in [0, T_{a}]$ where
\begin{equation}
T_a=\frac{2p_{0}}{\sqrt{4\epsilon^{2}+(\gamma_{0}+\gamma)^{2}}+(\gamma_{0}+\gamma)},
\end{equation}
we have
\begin{equation}
z_{t}\geq 1-(\max f(z))t\geq 1-2p_{0}.
\end{equation}
Therefore, if one makes a measurement on the system with $\sigma_{z}$, the probability of failure $\langle 1 |\rho_{t}|1\rangle=\frac{1-z_{t}}{2}\leq p_{0}$.
\end{proof}

\subsection{Proof of Corollary \ref{AC_improved_proposition1}}
\begin{proof}
When $p_{0}\leq \frac{1}{2}-\frac{\gamma_{0}+\gamma}{2\sqrt{4\epsilon^{2}+(\gamma_{0}+\gamma)^{2}}}$, from the proof of Theorem \ref{AC_Thereom1}, we know for $z\in [1-2p_{0}, 1]$,
\begin{equation}
\max f(z)=f(1-2p_{0})=4\epsilon\sqrt{p_{0}-p_{0}^{2}}+2(\gamma_{0}+\gamma)(1-p_{0}).
\end{equation}
Hence, if $t\in [0, T_{a}']$ where
\begin{equation}
T_{a}'=\frac{2p_{0}}{4\epsilon\sqrt{p_{0}-p_{0}^{2}}+2(\gamma_{0}+\gamma)(1-p_{0})},
\end{equation}
\begin{equation}
z_{t}\geq 1-[4\epsilon \sqrt{p_{0}-p_{0}^{2}}+2(\gamma_{0}+\gamma)(1-p_{0})]t\geq 1-2p_{0}.
\end{equation}
It is clear that the probability of failure $\langle 1 |\rho_{t}|1\rangle\leq p_{0}$.
\end{proof}

\subsection{Proof of Proposition \ref{AC_proposition2}}
\begin{proof}
When $H(t)=I_{z}$ and $\gamma_{t}=\gamma_{0}+\delta\gamma_{t}$, the state equation of the system in (\ref{AD_masterEq}) is
\begin{equation}
\left(%
\begin{array}{c}
  \dot{x}_{t} \\
  \dot{y}_{t} \\
  \dot{z}_{t} \\
\end{array}%
\right)
=\left(%
\begin{array}{ccc}
  -\frac{1}{2}(\gamma_{0}+\delta \gamma_{t}) & -1 & 0 \\
  1 & -\frac{1}{2}(\gamma_{0}+\delta \gamma_{t})  & 0 \\
  0 & 0 & -(\gamma_{0}+\delta \gamma_{t}) \\
\end{array}%
\right) \left(%
\begin{array}{c}
  x_{t} \\
  y_{t} \\
  z_{t} \\
\end{array}%
\right)+\left(%
\begin{array}{c}
  0 \\
  0 \\
  -(\gamma_{0}+\delta \gamma_{t}) \\
\end{array}%
\right),
\end{equation}
where $(x_{0}, y_{0}, z_{0})=(0, 0, 1)$. It is clear that
\begin{equation}\label{AD67}
\dot{z_{t}}=-(\gamma_{0}+\delta\gamma_{t})(1+z_{t})\geq -(\gamma_{0}+\gamma)(1+z_{t}).
\end{equation}
From (\ref{AD67}), we have
\begin{equation}
z_{t}\geq 2e^{-(\gamma_{0}+\gamma)t}-1.
\end{equation}
If $t\in [0, T_{a}'']$ where
\begin{equation}
T_{a}''=-\frac{\ln (1-p_{0})}{\gamma_{0}+\gamma},
\end{equation}
we have
\begin{equation}
z_{t}\geq 1-2p_{0}.
\end{equation}
That is, the probability of failure $\langle 1|\rho_{t}|1\rangle\leq p_{0}$.
\end{proof}

\subsection{Proof of Theorem \ref{AD_Theorem2}}
\begin{proof}
From the proof of Theorem \ref{AC_Thereom1}, we know
$$\dot{z_{t}}\geq -\max f(z)=-\sqrt{4\epsilon^{2}+(\gamma_{0}+\gamma)^{2}}-(\gamma_{0}+\gamma).$$
Now if the initial state $z_{0}\geq 1-2\alpha p_{0}$ and $t\in [0, (1-\beta)T_{a}]$, the system's state satisfies
\begin{equation}
z_{t}\geq z_{0}-(\sqrt{4\epsilon^{2}+(\gamma_{0}+\gamma)^{2}}+(\gamma_{0}+\gamma))(1-\beta) T_{a}\geq 1-2(1+\alpha-\beta)p_{0}.
\end{equation}
When $\alpha\leq \beta$, we have the following relationship
\begin{equation}
z_{t}\geq 1-2p_{0}.
\end{equation}
Hence, the probability of failure satisfies $\langle 1|\rho_{t}|1\rangle\leq p_{0}$.
\end{proof}

\subsection{Proof of Theorem \ref{PD_Theorem1}}
\begin{proof}
For a single qubit subject to (\ref{PD_masterEq})
%with $H(t)$ and phase damping decoherence with coupling strength $\gamma_{t}$, its dynamics can be described by the following equation
%\begin{equation}\label{PD_masterEq}
%\dot{\rho_{t}}=-i[H(t), \rho_{t}]+\gamma_{t}(\sigma_{z}\rho_{t} \sigma_{z}-\rho_{t})
%\end{equation}
when $H(t)=[1+\omega(t)] I_{z}+\epsilon_{x}(t)I_{x}+\epsilon_{y}(t)I_{y}$ ($|\omega(t)|\leq \omega$, $\sqrt{\epsilon_{x}^{2}(t)+\epsilon_{y}^{2}(t)}\leq
\epsilon$, $\omega\geq 0$ and $\epsilon >0$), $\gamma_{t}=\gamma_{0}+\delta \gamma_{t}$ ($|\delta \gamma_{t}|\leq \gamma$),
using (\ref{blochEq}), we have
\begin{equation}\label{lemmaBD}
\left(%
\begin{array}{c}
  \dot{x}_{t} \\
  \dot{y}_{t} \\
  \dot{z}_{t} \\
\end{array}%
\right)
=\left(%
\begin{array}{ccc}
  -2(\gamma_{0}+\delta \gamma_{t}) & -(1+\omega(t)) & \epsilon_{y}(t) \\
  1+\omega(t) & -2(\gamma_{0}+\delta \gamma_{t})  & -\epsilon_{x}(t) \\
  -\epsilon_{y}(t) & \epsilon_{x}(t) & 0 \\
\end{array}%
\right) \left(%
\begin{array}{c}
  x_{t} \\
  y_{t} \\
  z_{t} \\
\end{array}%
\right),
\end{equation}
where $C_{0}=x_{0}^{2}+y_{0}^{2}=1$.
Let $C_{t}=x_{t}^{2}+y_{t}^{2}$. We have
\begin{equation}\label{eq75}
\begin{array}{cc}
\dot{C_{t}}& =2x_{t}\dot{x}_{t}+2y_{t}\dot{y}_{t}=-4(\gamma_{0}+\delta\gamma_{t})(x_{t}^{2}+y_{t}^{2})+2z_{t}(\epsilon_{y}(t)x_{t}-\epsilon_{x}(t)y_{t})\\
& \geq -4(\gamma_{0}+\delta\gamma_{t})(x_{t}^{2}+y_{t}^{2})-2\epsilon \sqrt{1-(x_{t}^{2}+y_{t}^{2})}(|x_{t}|+|y_{t}|)\\
& \geq -4(\gamma_{0}+\delta\gamma_{t})(x_{t}^{2}+y_{t}^{2})-2\epsilon \sqrt{1-(x_{t}^{2}+y_{t}^{2})}\sqrt{2(x_{t}^{2}+y_{t}^{2})}
\end{array}
\end{equation}

Let $N_{t}=2(\gamma_{0}+\gamma)^{2}C_{t}^{2}-\epsilon^{2}C_{t}^{2}+\epsilon^{2}C_{t}$. We have
\begin{equation}\label{eq76}
%\begin{array}{cc}
2(\gamma_{0}+\delta\gamma_{t})C_{t}+\epsilon \sqrt{2C_{t}(1-C_{t})}
\leq \sqrt{2}\sqrt{4(\gamma_{0}+\delta\gamma_{t})^{2}C_{t}^{2}+2\epsilon^{2}C_{t}(1-C_{t})}
\leq2\sqrt{N_{t}}
%\end{array}
\end{equation}

Hence,
$
\dot{C_{t}}\geq -4\sqrt{\max N_{t}}
$. According to the definition of $N_{t}$, it is easy to verify the fact
\begin{equation}\label{maxN}
\max N_{t}=\begin{cases}2(\gamma_{0}+\gamma)^{2}, \  \ \ \text{when} \ 4(\gamma_{0}+\gamma)^{2}\geq \epsilon^{2};\\
\frac{\epsilon^{4}}{4\epsilon^{2}-8(\gamma_{0}+\gamma)^{2}}, \  \ \ \text{when} \ 4(\gamma_{0}+\gamma)^{2}< \epsilon^{2}.
\end{cases}
\end{equation}
%\begin{equation}
%C_{t}\geq 1-4\sqrt{2}(\gamma_{0}+\gamma)t.
%\end{equation}
If $t\in [0, T_{p}]$ where
\begin{equation}
T_{p}=\begin{cases}\frac{1-\bar{C}}{4\sqrt{2}(\gamma_{0}+\gamma)}, \  \ \ \text{when} \ 4(\gamma_{0}+\gamma)^{2}\geq \epsilon^{2};\\
\frac{(1-\bar{C})\sqrt{\epsilon^{2}-2(\gamma_{0}+\gamma)^{2}}}{2\epsilon^{2}}, \  \ \ \text{when} \ 4(\gamma_{0}+\gamma)^{2}< \epsilon^{2},
\end{cases}
\end{equation}
we have
$C_{t}\geq \bar{C}$.

\end{proof}

\subsection{Proof of Corollary \ref{PD_proposition1}}
\begin{proof}
When $\epsilon^{2}=2(\gamma_{0}+\gamma)^{2}$, from (\ref{eq75}) and (\ref{eq76}), we have
\begin{equation}
\dot{C_{t}}\geq -4\sqrt{\epsilon^{2}C_{t}}.
\end{equation}
It is easy to obtain the following relationship
\begin{equation}
2\text{d}\sqrt{C_{t}}\geq -4\epsilon \text{d}t,
\end{equation}
\begin{equation}
\sqrt{C_{t}}\geq 1-2\epsilon t.
\end{equation}
If $t\in [0, T_{p}']$ where
\begin{equation}
T_{p}'=\frac{1-\sqrt{\bar{C}}}{2\sqrt{2}(\gamma_{0}+\gamma)},
\end{equation}
we have
$C_{t}\geq \bar{C}.$
\end{proof}

\subsection{Proof of Proposition \ref{PD_proposition2}}
\begin{proof}
When $H(t)=I_{z}$ and $\gamma_{t}=\gamma_{0}+\delta\gamma_{t}$, using (\ref{blochEq}) and (\ref{PD_masterEq}), we have
\begin{equation}\label{lemmaBD2}
\left(%
\begin{array}{c}
  \dot{x}_{t} \\
  \dot{y}_{t} \\
  \dot{z}_{t} \\
\end{array}%
\right)
=\left(%
\begin{array}{ccc}
  -2(\gamma_{0}+\delta \gamma_{t}) & -1 & 0 \\
  1 & -2(\gamma_{0}+\delta \gamma_{t})  & 0 \\
  0 & 0 & 0 \\
\end{array}%
\right) \left(%
\begin{array}{c}
  x_{t} \\
  y_{t} \\
  z_{t} \\
\end{array}%
\right),
\end{equation}
where $C_{0}=x_{0}^{2}+y_{0}^{2}=1$. It is clear that
\begin{equation}
\dot{C_{t}}=2x_{t}\dot{x}_{t}+2y_{t}\dot{y}_{t}=-4(\gamma_{0}+\delta\gamma_{t})(x_{t}^{2}+y_{t}^{2})\geq -4(\gamma_{0}+\gamma)C_{t}.
\end{equation}
Hence,
\begin{equation}
C_{t}\geq e^{-4(\gamma_{0}+\gamma)t}.
\end{equation}
If $t\in [0, T_{p}'']$ where
\begin{equation}
T_{p}''=-\frac{\ln \bar{C}}{4(\gamma_{0}+\gamma)},
\end{equation}
we have
$C_{t}\geq \bar{C}$.
\end{proof}

\subsection{Proof of Theorem \ref{DD_Theorem}}
\begin{proof}
For a single qubit system subject to (\ref{DD_masterEq}),
%with $H(t)$ and depolarizing decoherence with coupling strength $\gamma_{t}$, the dynamics can be described by the %following equation
%\begin{equation}
%\dot{\rho_{t}}=-i[H(t), \rho_{t}]+\gamma_{t}(\sigma_{x}\rho_{t} \sigma_{x}-\rho_{t})+\gamma_{t}(\sigma_{y}\rho_{t} %\sigma_{y}-\rho_{t})+\gamma_{t}(\sigma_{z}\rho_{t} \sigma_{z}-\rho_{t})
%\end{equation}
when $H(t)=[1+\omega(t)] I_{z}+\epsilon_{x}(t)I_{x}+\epsilon_{y}(t)I_{y}$ ($\sqrt{\epsilon_{x}^{2}(t)+\epsilon_{y}^{2}(t)}\leq
\epsilon$, $\epsilon >0$, $|\omega(t)|\leq \omega$ and $\omega\geq 0$), $\gamma_{t}=\gamma_{0}+\delta \gamma_{t}$ ($|\delta \gamma_{t}|\leq \gamma$),
using (\ref{blochEq}), we have
\begin{equation}
\left(%
\begin{array}{c}
  \dot{x}_{t} \\
  \dot{y}_{t} \\
  \dot{z}_{t} \\
\end{array}%
\right)
=\left(%
\begin{array}{ccc}
  -4(\gamma_{0}+\delta \gamma_{t}) & -1 & 0 \\
  1 & -4(\gamma_{0}+\delta \gamma_{t})  & 0 \\
  0 & 0 & -4(\gamma_{0}+\delta \gamma_{t}) \\
\end{array}%
\right) \left(%
\begin{array}{c}
  x_{t} \\
  y_{t} \\
  z_{t} \\
\end{array}%
\right),
\end{equation}
where $P_{0}=x_{0}^{2}+y_{0}^{2}+z_{0}^{2}=1$ and $R_{t}=\text{tr}(\rho_{t}^{2})=x_{t}^{2}+y_{t}^{2}+z_{t}^{2}.$ It is clear that
\begin{equation}
\dot{R_{t}}=-8(\gamma_{0}+\delta \gamma_{t})P_{t}\geq -8(\gamma_{0}+\gamma)R_{t}.
\end{equation}
Hence,
\begin{equation}
R_{t}\geq e^{-8(\gamma_{0}+\gamma)t}
\end{equation}
If $t\in [0, T_{d}]$ where
\begin{equation}
T_{d}=-\frac{\ln (2\bar{P}-1)}{8(\gamma+\gamma_{0})},
\end{equation}
we have
$P_{t}\geq \bar{P}.$
\end{proof}

%%%%%%%%%%%%%%%%%%%%%%%%%%%%%%%%%%%%%%%%%%%%%%%%%%%%%%%%%%%%%%%%%%%%%%%%%%%%%%%%
\section{CONCLUSIONS}\label{Sec5}
Control design for quantum systems with uncertainties is an important task. This paper has proposed a sampled-data design approach for a single qubit with uncertainties. Both closed and Markovian open quantum systems are investigated, and uncertainties in the
system Hamiltonian and uncertainties in the coupling strength of the system-environment interaction are analyzed. Several physically meaningful performance indices including fidelity, coherence and purity are used to define the required robustness and several sufficient
conditions on the relationships between related parameters in the control system are
established to guarantee such
robustness. The robust control law can be designed offline and then be used online on the single qubit system with uncertainties. Future work will include the extension of these sampled-data control approaches to other finite dimensional quantum systems and the development of practical applications of the proposed method.

%unified framework for protecting quantum information \textbf{We can
%also select a decoherence-free subspace \cite{Lidar et al 1998},
%\cite{Kwiat et al 2000} as a sliding mode} \cite{Viola and Lloyd
%1998}

%%%%%%%%%%%%%%%%%%%%%%%%%%%%%%%%%%%%%%%%%%%%%%%%%%%%%%%%%%%%%%%%%%%%%%%%%%%%%%%%
%\section{ACKNOWLEDGMENTS}

%The first author would like to thank Bo Qi and Hao Pan for helpful
%suggestions

%%%%%%%%%%%%%%%%%%%%%%%%%%%%%%%%%%%%%%%%%%%%%%%%%%%%%%%%%%%%%%%%%%%%%%%%%%%%%%%%

%References are important to the reader; therefore, each citation must be complete and correct. If at all possible, references should be commonly available publications.

\end{document}